\newif\ifnotes
\newif\ifsubmission
\newif\ifextendedabstract

\ifsubmission
    \documentclass{llncs}
\else
    \documentclass[11pt]{article} 
\fi

\usepackage[utf8]{inputenc}
\ifsubmission\else\usepackage[pdftex,left=1in,top=1in,bottom=1in,right=1in]{geometry}\fi
\usepackage{amsmath, amssymb, dsfont, mathrsfs}
\ifsubmission \else \usepackage{amsthm} \fi
\usepackage{braket}
\usepackage{graphicx}
\usepackage{caption}
\usepackage{bm}
\usepackage[dvipsnames]{xcolor}
\usepackage[framemethod=tikz]{mdframed} 
\usepackage[colorlinks=true,citecolor=blue,linkcolor=blue]{hyperref}
\usepackage[nameinlink, noabbrev, capitalize]{cleveref}
\usepackage{physics}
\usepackage{verbatim} 
\usepackage{mathtools}
\usepackage{enumitem}
\hypersetup{
    colorlinks,
    linkcolor={red!50!black},
    citecolor={blue!50!black},
    urlcolor={blue!80!black}
}
\usepackage{relsize}

\mathchardef\mhyphen="2D

\definecolor{comment}{RGB}{176, 224, 128}

\newcommand{\F}{\mathds{F}}

\newcommand{\cX}{\mathcal{X}}
\newcommand{\cY}{\mathcal{Y}}

\newcommand{\sk}{\mathsf{sk}}

\newcommand{\samp}{\leftarrow}
\newcommand{\reg}{\mathsf{R}}
\newcommand{\regi}[1]{\reg_{\mathsf{#1}}}

\ifsubmission \else
    \newtheorem{theorem}{Theorem}
    
    \newtheorem{claim}{Claim}
    \crefname{claim}{Claim}{Claims}
    \newtheorem{lemma}{Lemma}
    
    \newtheorem{corollary}{Corollary}
    
    \newtheorem{definition}{Definition}
    \newtheorem{remark}{Remark}
    
\fi

\newtheorem{construction}{Construction}

\newcommand{\poly}{\mathsf{poly}}

\newcommand{\zo}{\{0, 1\}}

\newcommand{\negl}{\mathsf{negl}}

\newcommand{\Sim}{\mathcal{S}}

\newcommand{\Adv}{\adv}

\newcommand{\Sign}{\mathsf{Sign}}

\ifnotes

\fi

\let\originalleft\left
\let\originalright\right
\renewcommand{\left}{\mathopen{}\mathclose\bgroup\originalleft}
\renewcommand{\right}{\aftergroup\egroup\originalright}

\newcommand{\cert}{\mathsf{cert}}

\newcommand\restr[2]{{
  \left.\kern-\nulldelimiterspace 
  #1 
  \vphantom{\big|} 
  \right|_{#2} 
  }}

\newcommand{\authnote}[3]{\textcolor{#3}{[{\footnotesize {\bf #1:} { {#2}}}]}}
\newcommand{\alper}[1]{\ifnotes \authnote{Alper}{#1}{blue} \fi}
\newcommand{\justin}[1]{\ifnotes \authnote{Justin}{#1}{ForestGreen} \fi}

\newcommand{\PQSig}{\mathsf{PQSig}}

\newcommand{\Gen}{\mathsf{Gen}}

\newcommand{\Prove}{\mathsf{Prove}}
\newcommand{\Verify}{\mathsf{Verify}}
\newcommand{\Del}{\mathsf{Del}}
\newcommand{\DelVer}{\mathsf{DelVer}}

\newcommand{\FS}{\mathsf{FS}}

\newcommand{\ch}{\mathsf{ch}}

\newcommand{\verkey}{\mathsf{vk}}
\newcommand{\delkey}{\mathsf{dk}}
\newcommand{\secpar}{\lambda}
\newcommand{\Accept}{\mathsf{Accept}}

\newcommand{\lang}{\mathcal{L}}
\newcommand{\Hyb}{\mathsf{Hyb}}
\newcommand{\adv}{\mathsf{Adv}}
\newcommand{\aux}{\mathsf{aux}}

\newcommand{\sigcd}{\mathsf{Sig\text{-}CDen}}
\newcommand{\nizkcd}{\mathsf{NIZK\text{-}CDen}}
\newcommand{\fscd}{\mathsf{FS\text{-}CDen}}

\newcommand{\sigcdenexp}{\mathsf{Sig\text{-}CDen\text{-}Exp}}
\newcommand{\nizkcdenexp}{\mathsf{NIZK\text{-}CDen\text{-}Exp}}

\newcommand{\concat}{\|}
\newcommand{\bbN}{\mathbb{N}}
\newcommand{\tracedist}{\mathsf{TD}}
\newcommand{\Func}{\mathsf{Func}}

\title{How to Delete Without a Trace:\\Certified Deniability in a Quantum World}

 \author{Alper \c{C}akan\thanks{Carnegie Mellon University. \texttt{acakan@cs.cmu.edu}.} \and Vipul Goyal\thanks{NTT Research \& Carnegie Mellon University.  \texttt{vipul@vipulgoyal.org}} \and Justin Raizes\thanks{Carnegie Mellon University \texttt{jraizes@andrew.cmu.edu}.}}

\date{}

\begin{document}
\maketitle
\ifextendedabstract\else
\begin{abstract}
    Is it possible to comprehensively destroy a piece of quantum information, so that \emph{nothing} is left behind except the memory of whether one had it at one point? For example, various works, most recently Morimae, Poremba, and Yamakawa (TQC '24), show how to construct a signature scheme with certified deletion where a user who deletes a signature on $m$ cannot later produce a signature for $m$. However, in all of the existing schemes, even after deletion the user is still able keep irrefutable evidence that $m$ was signed, and thus they do not fully capture the spirit of \emph{deletion}.

    In this work, we initiate the study of \emph{certified deniability} in order to obtain a more comprehensive notion of deletion. Certified deniability uses a simulation-based security definition, ensuring that any information the user has kept after deletion could have been learned without being given the deleteable object to begin with; meaning that deletion leaves no trace behind! We define and construct two non-interactive primitives that satisfy certified deniability in the quantum random oracle model: signatures and non-interactive zero-knowledge arguments (NIZKs). As a consequence, for example, it is not possible to delete a signature/NIZK and later provide convincing evidence that it used to exist. Notably, our results utilize uniquely quantum phenomena to bypass Pass's (CRYPTO '03) celebrated result showing that deniable NIZKs are impossible even in the random oracle model.
\end{abstract}
\fi

\ifextendedabstract
\else
    \section{Introduction}

In their seminal work, Broadbent and Islam~\cite{TCC:BroIsl20} observed that it is possible to enforce the deletion of secret information via quantum mechanics, introducing the notion of encryption with certified deletion. The notion of certified deletion was subsequently generalized to a wide variety of cryptographic primitives, e.g.~\cite{C:HMNY22,EC:HKMNPY24,ITCS:Poremba23,TCC:AnaPorVai23,EC:AKNYY23,C:BarKhu23,EC:BGKMRR24,C:BR24}.\footnote{The notion of cryptography with certified deletion has also been referred to as ``revocable cryptography''.}
Following this line of work, Morimae, Poremba, and Yamakawa~\cite{TQC:MPY24} formally defined the notion of revocable signatures\footnote{We note that while it was not formally studied before \cite{TQC:MPY24}, it was folklore that revocable signatures follow from public-key quantum money. \cite{TQC:MPY24} construct revocable signatures from one-way functions.}. A revocable signature scheme enables a sender to sign a message $m$ in such a way that it may later verify that the signature has been destroyed, in the sense that the user who previously received the signature can no longer produce a valid signature for $m$. 

We argue that although the existing notion of revocable signatures is well-motivated, it loses some of the spirit of deletion. Their definition requires that after the user produces a valid certificate, they cannot produce a signature that is recognized as valid \emph{by the honest verification algorithm}. On the other hand, the user may be able to use the leftover information to prove to a third party that the message was signed \emph{at some point}. Indeed, all the existing constructions, including that of \cite{TQC:MPY24}, suffer from such an attack.\footnote{We note that despite this weakness, existing constructions still have applications in settings where the third-party verifier can be trusted to be honest, such as building access cards.}

This crack in the armor of their definition comes from a custom tailoring of the certified deletion paradigm to the specific security guarantees of the target cryptographic primitive - a trait shared by prior works.
In light of this, we propose a more 
comprehensive
philosophy for certified deletion:
\begin{quote}
    \begin{center}
        \emph{Once a user deletes the delegated information, it should be as if they never received it in the first place.}
    \end{center}
\end{quote}

This philosophy has close connections to a natural notion called \emph{deniability} in cryptography,  which has been well-studied in the classical setting~\cite{STOC:DolDwoNao91,C:CDNO97,STOC:DwoNaoSah98,C:Pass03} and has recently received attention in the quantum setting~\cite{STOC:ColGolVaz22}.
Although there are good solutions for interactive deniability, non-interactive deniability is notoriously hard in the classical setting.
Solutions for signatures make heavy compromises such as sacrificing public signature verifiability or requiring well-synchronized global clocks~\cite{EC:JakSakImp96,SP:HulWeb21,AC:AruBonCla22,POPETS:BCGJT23}, while deniable non-interactive zero knowledge is simply impossible even with a random oracle or common reference string~\cite{C:Pass03}.

\paragraph{Certified Deniability for Signatures and NIZKs.} 
In this work, we explore the formalization of our new certified deletion philosophy, which we call \emph{certified deniability}. We define and construct certified deniability for two natural non-interactive primitives: signatures and non-interactive zero knowledge (NIZK).
In a signature or NIZK with certified deniability, any user may publicly verify that a message was signed or that some statement $x$ is true. However, once a signature/NIZK is deleted, the party that had it will retain only their memory of whether they decided to accept or reject it. 
As a consequence, they are unable to convince a skeptical external verifier of the truth of the statement, even if the verifier is willing to accept an alternative proof outside the constraints of the original protocol.

\paragraph{Application: Deniable Contracts.} As a motivating example, consider two parties Alice and Bob who sign a contract. Their signatures should be publicly verifiable so that if either party does not uphold their end, the contract can be publicly verified in court. However, after the contract has ended, neither party should be able to provide evidence that, e.g. Alice once paid for a particular product or service from Bob. To ensure this, at the termination of the contract, the parties can both delete the other's signature.  

\justin{I think contracts are much more motivating than employment. Deniable employment seems very idealized, since you could just ask for details that only an employee would know.}

\subsection{Our Results.}

As the main contributions of this work, we lay the definitional groundwork for signatures and NIZKs with certified deniability and study what assumptions enable them. 

\paragraph{Simulation-Based Deniability.} 
To define certified deniability, we use the simulation paradigm to capture the comprehensive deletion philosophy mentioned before. 
We require that the state of the adversary after producing a valid deletion certificate can be simulated \emph{without} having received a signature (or NIZK) in the first place. 
This gives us a guarantee of the form \emph{``anything an adversary could learn from a deleted signature/NIZK, they could learn without it.''}


To achieve this notion for signatures and NIZKs, we augment the Fiat-Shamir transform to add certified deniability in the quantum random oracle model (QROM).

\begin{theorem}[Informal; \cref{thm:delsigexists}, \cref{thm:delnizkexists}]
    There exist NIZKs with certified deniability in the QROM. Furthermore, if one-way functions exist, then there exist signatures with certified deniability in the QROM.
\end{theorem}

Notably, this result bypasses Pass's impossibility for deniable NIZKs in the random oracle model. This can be seen as a unique application of quantum mechanics to force a user to ``forget'' information in a quantifiable way. For more discussion, see the technical overview.

\paragraph{Evidence Against Plain Model Constructions.} 
Since our results are in the QROM, a natural next question is whether we can hope to achieve certified deniability (of either variety) in the \emph{plain} model. Unfortunately, it appears that there are significant barriers to doing so. We show that if the security proof treats the adversary as a black-box, then it cannot hope to show either of the above notions.

\begin{theorem}[Informal]
   There is no signature/NIZK with certified deniability in the plain model with a security proof that makes black-box use of the adversary.
\end{theorem}

Thus, any valid proof of security for a plain-model construction must be non-black-box. Although there have been many improvements in non-black-box techniques in the past decade, e.g.~\cite{STOC:BitKalPan18,STOC:BitShm20}, 
the technique we use for the barrier is particularly amenable to obfuscation, which has also seen vast improvements recently~\cite{STOC:JaiLinSah21,EC:JaiLinSah22,C:JLLW23}.
We leave a more definitive answer to the question of certified deniability in the plain model to future work.

    \section{Technical Overview}

\paragraph{Evidence-Collection Attacks.} 
We begin by discussing how an adversarial verifier can retain irrefutable evidence of a signature even after it has been revoked in the existing revocable signature constructions. We will consider MPY24's construction as an example, though the same approach works for other existing constructions (such as the public-key quantum money based approach).
MPY24 begins by defining and constructing a new primitive called ``2-tier tokenized signatures''. 
This primitive enables a simple construction for revocable signatures that can sign only the message ``$0$''.
They extend this construction to support multiple messages by generating a fresh key pair $(\sk_T, \verkey_T)$, together with a token $\ket{T}$, whenever the signer wishes to sign a new message $m$. Then, the signer signs $\verkey_T \concat m$ under their global signing key to associate the new token with $m$.

The problem lies in this final step. Every signature on a message $m$ consists of a \emph{classical} signature $\sigma$ on $\verkey_T \concat m$, together with the revocable token $\ket{T}$. However, ignoring $\ket{T}$ completely, $\sigma$ is already irrefutable proof that the signer signed $m$!

\subsection{Definitions: Simulation-Style} 

We provide an in-depth discussion of our definitions and their merits. As mentioned previously, we use a simulation-style definition, which offers a robustness even against classes of attacks that have not yet been explicitly considered. 
As demonstrated by the evidence-collection on MPY24's revocable signatures, such robustness is invaluable for avoiding subtleties that creative attackers might take advantage of.

\paragraph{Signatures.}
In certified deniability for signatures, we consider the following scenario that might happen in the real world:
\begin{enumerate}
    \item The adversarial verifier $V^*$ receives a public verification key $\verkey$.
    \item $V^*$ receives signatures on some set of messages $M$ over a period of time. At some point, it decides to delete a signature $\sigma_{m^*}$ on some $m^*$ by outputting a certificate $\cert$, which it sends to the signer together with $m^*$.\footnote{In the technical sections, we actually require that the adversary declares that it will delete $\sigma_{m^*}$ when it queries for a signature on $m^*$, instead of deciding later. We can remove this requirement using standard complexity leveraging techniques.}
    \item The signer decides whether to accept or reject $\cert$.
\end{enumerate}
To transform this scenario into a security experiment $\sigcdenexp_{V^*}(\sk, \verkey)$, we replace $V^*$'s output with $\bot$ if the signer rejects $\cert$. 
Additionally, the experiment outputs the list of messages $M\backslash\{m^*\}$ whose signatures $V^*$ received, \emph{not including $m^*$}. 

Certified deniability requires that there exists a simulator $\Sim$ which receives $\verkey$ and access to whichever signatures it wants, then produces a view that is indistinguishable from $\sigcdenexp_{V^*}(\sk, \verkey)$ \emph{even when the list of messages $\Sim$ queries for is included in the output}:
\begin{gather*}
   \left\{(\verkey, \mathsf{Sig\text{-}CDen\text{-}Exp}(\sk, \verkey))\right\}_{(\sk, \verkey) \gets \Gen(1^\secpar)}
    \approx
    \left\{(\verkey, M_\Sim, \Sim^{\Sign(\sk, \cdot)}  \right\}_{(\sk, \verkey) \gets \Gen(1^\secpar)}
\end{gather*}
Since the real adversary gets a ``free signature'' on $m^*$ that is not recorded in $M\backslash\{m^*\}$, but all of $\Sim$'s queries are recorded, $\Sim$ must be able to produce its view \emph{without} receiving a signature for $m^*$.\footnote{Alternatively, we could try to give $\Sim$ a restricted signing oracle that never signs $m^*$. However, $m^*$ is not determined until late in the experiment, so such a restricted signing oracle is not well-defined.}

In particular, this definition means that if $\adv$ never asks for any signatures except $m^*$, then $\Sim$ must be able to simulate its view without receiving any signatures at all! Thus, the definition requires that even the \emph{number} of signatures which have been given out is deleted when all of them are deleted.

\paragraph{NIZKs.} Certified deniability for NIZKs is defined similarly. In the real execution $\nizkcdenexp$, the prover uses a witness $w$ to prove the truth of some statement $x$. The simulator must reproduce the (adversarial) verifier's view given only $x$, without a witness $w$.
The reader may notice that this definition seems almost identical to the standard zero knowledge definition. Indeed, Pass~\cite{C:Pass03} points out that zero knowledge is inherently deniable in the plain model; it is indistinguishable whether the transcript given to you was the result of an honest execution, or the result of running the simulator\footnote{In the case where the statement is hard to decide, a simulator could even ``prove'' a false statement.}.
Of course, NIZKs do not exist in the plain model~\cite{JC:GolOre94} and thus to achieve NIZKs, we must move to either the random oracle model or common reference string (CRS) model. Therefore, it is a moot point to consider the deniability of NIZKs in the plain model, and we need to consider its deniability in the random oracle model or the common reference string model.

\paragraph{Certified Deniability in the QROM.} Traditionally for NIZKs in the random oracle model, the simulator is allowed to choose the random oracle, or similarly it is allowed to choose common reference string in the CRS model. However, Pass points out that in the real world, the random oracle is fixed once and for all. It is not realistic to believe that an arbitrary prover backdoored the random oracle. So, they could not have run the simulator.

Following Pass, we define certified deniability in the QROM so that the simulator does not have the ability to choose the random oracle. Specifically, in the ideal world, the simulator has access to a truly random oracle $H$, which also appears in the experiment output:
\[
    \{(H, \nizkcd^H(x,w))\}_{H \samp \mathcal{H}}
    \approx 
    \{(H, \Sim^H(x))\}_{H \samp \mathcal{H}}
\]
where $\mathcal{H}$ is the set of all functions $\zo^{p(\lambda)} \to \zo^{q(\lambda)}$.

In the classical setting, if the simulator were to internally pretend that some $H(y) = v$, it would be immediately caught by any distinguisher who queried the real $H$ on $y$. 
Pass utilizes the non-programmability enforced by this definition to show that (in the classical setting) deniable NIZKs are in fact \emph{impossible} even with a random oracle.

\paragraph{The Power of Forgetting.} \justin{Not certain what the comment "less clear" on the title means.}
We show that it is possible to avoid this issue in the quantum setting. 
Quantum mechanics offers the intriguing capability to force the adversary to \emph{forget} information it queried on when it produces a valid deletion certificate.
This enables the simulator to internally pretend that $H(y) = v$ without later being detected by the distinguisher querying the real $H$ on $y$, since the distinguisher never learns $y$.

We remark that in real world, where $H$ is heuristically implemented, the third party does not actually believe that the implementation changed to be $H(y) = v$. 
However, we can heuristically say that if they do not know $y$, then they do not know anything about $H(y)$, so they have no evidence against $H(y) = v$ in the actual implementation either.

\paragraph{On Active vs Passive Primitives.} 
Basing security definition for certified deletion on the simulation paradigm seems to have much stricter requirements. However, this is not the case for all primitives. In particular, passive primitives - such as in encryption or commitments - where the adversary is not intended to learn anything before deleting, naturally satisfy a simulation-style definition. For example, in encryption (the original primitive which certified deletion was added to), the adversary should not be able to tell the difference between an encryption of $0$ and an encryption of $1$ after deletion, even if after deletion it obtains the secret key or becomes computationally unbounded. In this case, this simulator can simply encrypt a fixed message, since it will anyway become indistinguishable after deletion.

On the other hand, the simulation paradigm introduces new security guarantees for active primitives, where the adversary is intended to learn something before it deletes. For example, in a signature, it learns that the signer endorses a particular message. Other ``active''  primitives with certified deletion include zero knowledge~\cite{C:HMNY22} and obfuscation~\cite{EC:BGKMRR24}. 

\paragraph{Before-the-Fact Coercion.} \cite{STOC:ColGolVaz22} previously introduced the idea of ``before-the-fact'' coercion for deniable encryption. This notion considers a coercer who approaches a victim before the victim computes a ciphertext (e.g. to cast a vote), then forces them to encrypt a particular message by requiring them to produce evidence that they encrypted that message. Even with this modification, the coercer should not be able to identify whether the encryptor encrypted a desired message $m$, or something else. Coledangelo, Goldwasser, and Vazirani showed that this notion is classical impossible, but possible to achieve using quantum techniques.

Signatures and NIZKs with certified deniability are naturally deniable even against before-the-fact coercion of the verifier. Even if the victimized verifier is forced to use a special device or obfuscated program to collect evidence about the received signature/NIZK, this evidence is rendered useless after the verifier produces a valid certificate. 
Classically, this notion is not possible even in the random oracle model, since the auxiliary program could record any ROM input-output behavior when run by the victim, enabling the coercer to check its correctness later. However, in the quantum setting, any such record might prevent the verifier from outputting a valid certificate since the certificate generation requires ``forgetting'' QROM queries. 
Furthermore, if the victim is forced to collect evidence anyway by not giving a valid certificate, then the authority can at least identify that they have been subject to coercion!

\subsection{Constructions}

We now give an overview of our constructions, starting with signatures.

\paragraph{Starting Point.} We start by recalling how \cite{TQC:MPY24} constructs revocable signatures for the message ``$0$''. As a base, they sample a random pair of values $(x_0, x_1)$ and a random phase $c\gets\{0,1\}$. Then, they compute one-way function images $f(x_0)$ and $f(x_1)$ and output 
\[
    \Sign(\sk, f(x_0) \concat f(x_1)),\quad \ket{\psi}\coloneqq \ket{x_0} + (-1)^c\ket{x_1}
\]
Given this, a verifier can measure $\ket{\psi}$ and check the result matches either $f(x_0)$ or $f(x_1)$. This operation can be done coherently to avoid disturbing the state, enabling the verifier to delete after it verifies.
The deletion certificate is obtained by measuring $\ket{\psi}$ in the Hadamard basis, which results in a vector $d$ such that $c = d\cdot (r_0 \oplus r_1)$. Using knowledge of $c$, $r_0$, and $r_1$, the signer can verify this certificate.

\cite{TQC:MPY24} shows that if the adversary cannot produce both such a $d$ and one of $x_0$ or $x_1$, except with probability $1/2$. Then, they show that repeating the scheme in parallel amplifies the difficulty of this task to a $\negl$ success rate; no adversary can obtain both a Hadamard basis measurement of every index and a string containing one element from every pair $(x_0, x_1)$. To extend the scheme to other messages $m$, they instead sign $m$ together with all of the one-way function images.

\paragraph{Attempt at a Fix.} As discussed previously, the weakness in MPY24's construction is that it directly signs the message $m$ under the global verification key. A natural attempt to eliminate this weakness is to instead sign something that is binding on $m$, but appears independent after deletion. For example, if we could somehow enforce that the adversary ``forgets'' a random string $r$ after deletion, then the signer could instead give out a signature on a random oracle image $H(r\concat m)$, along with $r$ and $m$.

This approach gives us a template for enforcing the adversary's forgetfulness. If we were to use $x_0$ and $x_1$ as the values of $r$, then we could give out $\ket{\psi}$ and a signature on $H(x_0\concat m)\concat H(x_1\concat m)$ (as opposed to a signature on $m\concat f(x_0) \concat f(x_1)$). Of course, this only ensures $1/2$ ``forgetfulness''. To amplify, we could secret share $m$ using the additive $\secpar$-of-$\secpar$ threshold secret sharing scheme, obtaining $m_1, \dots, m_\secpar$ such that $m$ $m_1\oplus \dots \oplus m_\secpar = m$, and sign each share $m_i$ using the modified scheme:
\[
    \Sign\left(\sk, H\left(r^i_0\concat m_i\right) \concat H\left(r^i_1\concat m_i\right)\right),\ m_i,\ \ket{\psi_i}\coloneqq \ket{r^i_0} + (-1)^{c_i}\ket{r^i_1}
\]
If the adversary forgets even a single $r_i^1$, they would be unable to verify $m_i$. However, changing any $m_i$ changes the signed message $m$.

\paragraph{Proof Technique: Forgetful Local Programming.} Although this construction turns out to not fully satisfy certified deniability, it will be helpful in demonstrating our new technique: forgetful local programming. The simulator will require a signature $\sigma_0$ on any message, say on ``0''(note that this leaks the number of messages signed, if not which ones). It will attempt to locally convince the verifier that this signature is actually for $m$.

To do so, it picks a random index $i$ and computes $m_i' = m_i \oplus m$. If $m_i'$ were swapped for $m_i$, then the secret-shared message becomes $(m_1\oplus \dots \oplus m_\secpar) \oplus m = 0 \oplus m = m$. To create this change, the simulator creates a modified oracle $H'$, which swaps the behavior of $r^i_b \concat m_i$ with $r^i_b \concat m_i'$, i.e. 
\begin{gather*}
    H'(r^i_b \concat m_i') \coloneqq H(r^i_b \concat m_i)
    \\
    H'(r^i_b \concat m_i) \coloneqq H(r^i_b \concat m_i')
\end{gather*}
Then, it runs the adversarial verifier $V^*$ using $H'$ and $\sigma_0$. If $H'$ were the true oracle, then $\sigma_0$ is actually a signature on $m$.

The crux of the argument comes down to showing that any third party distinguisher cannot distinguish whether it has access to $H$ or $H'$. Since $\sigma_0$ is a signature on $0$ under $H$, or a signature on $m$ under $H'$, it cannot tell which message $\sigma_0$ was originally associated with. To argue the crux, we will show that if $V^*$ produces a valid certificate, then the distinguisher cannot find either $r^i_0$ or $r^i_1$.

\paragraph{The Need for a New Approach.} Unfortunately, this approach only achieves $1/\poly(\secpar)$ security.\footnote{Actually, this cannot be black-box reduced to the hardness of finding $r^i_0$ or $r^i_1$ discussed above, because of a large loss that occurs when extracting QROM queries. This problem can be fixed by outputting the whole signature as the certificate and doing some additional technical work, though we omit the details since it is subsumed by our later construction.} The issue is that the adversary could simply guess $i$, and is correct with inverse polynomial probability. If it deletes the other indices honestly, then it can keep one of $r_0^i$ and $r_1^i$ with probability $1/2$ even while providing a valid deletion certificate. Then, given $r^i_b$, the distinguisher could immediately detect the local reprogramming. Without the forgetful local programming technique, the tools of constructing certifiably deniable NIZKs becomes very close to what is available in the classical setting, where we know the task is impossible.


One might hope that by increasing the number of values in superposition, e.g. by constructing $\sum_{i}(-1)^{c_i}\ket{r_i}$ instead of $\ket{r_0} \pm \ket{r_1}$, the likelihood that the adversary can keep an $r_i$ while simultaneously producing a valid deletion certificate can be decreased. While this is true, a superposition over many $r_i$ requires a proportionate number of signed random oracle images $H(r_i\concat m)$. Increasing the ``forgetfulness'' of the adversary to overwhelming probability would blow up the size of the signature super-polynomially.

Another issue is that this approach does not allow the \emph{number} of signatures which have been given out to be deleted. Although the simulator can attempt to locally pretend that a signature on $m$ is actually a signature on $m'$ by locally reprogramming the oracle, it needs to have a signed oracle image $H(r\concat m)$ to begin with.

\paragraph{Idea 1: Subspace States.} To increase the ``forgetfulness'' of the adversary beyond polynomial factors, we turn to subspace states~\cite{AC12}. A subspace state $\ket{A}\propto \sum_{a\in A}$ is a uniform superposition over elements of an $\secpar/2$-dimensional subspace $A$ of $\F_2^\secpar$. 
Using them, we can avoid the exponential blow-up by signing the random oracle images in superposition:
\[
    \ket{\sigma_m} \coloneqq \sum_{a\in A\backslash\{0\}} \ket{a} \otimes \ket{\Sign(\sk, H(a\concat m))}
\]

The verifier can check such a state by (1) coherently running the signature verification procedure in the computational basis and (2) coherently checking that the signed message matches $H(a\concat m)$. To delete the signature, they can simply return the whole state. 
If the signing algorithm is deterministic, then the signer can check the certificate by uncomputing the signature and random oracle image, then checking that the certificate is now $\ket{A}$ using a projection onto $\ket{A}$. We note that it is possible to make the signing algorithm deterministic by coherently deriving its randomness from a PRF evaluated at $a$, e.g. $H(k\concat a)$ for random $k$.

We argue that this deletion check enforces ``forgetfulness''. Due to the direct product hardness property \cite{Quantum:BS23}, we know that given a random subspace state $\ket{A}$, it is hard to find both a vector in $A\backslash\{0\}$ and a vector in $A^\perp\backslash\{0\}$, even given access to an oracle which decides membership in $A$ and $A^\perp$. 
The membership oracle can be used to check if the returned certificate indeed contains an intact copy of $\ket{A}$. If the signer is able to recover $\ket{A}$ from the certificate, then it could obtain a vector in $A^\perp$ by measuring it in the Hadamard basis. 
Whenever this happens, the verifier cannot also remember \emph{any} vector in $A$, other than $0$. Thus, we can use the forgetful local programming technique to obtain $\negl(\secpar)$ security loss.

\paragraph{Idea 2: Fiat-Shamir in Superposition.} This still leaves the issue of leaking the number of signatures. To solve this, we use the Fiat-Shamir paradigm in superposition. 
Fiat-Shamir transforms a sigma protocol\footnote{A sigma protocol is a 3-message public-coin argument of knowledge which is zero-knowledge against an honest verifier whose second message is known ahead of time.} into a signature scheme. The signer samples a random secret key $\sk$ and gives out $\verkey = f(\sk)$ as the verification key, where $f$ is a one-way function.
To sign a message me, the signer computes a sigma protocol proving knowledge of an $\sk$ matching $\verkey$. Fiat-Shamir uses the first message $s_1$ of the sigma protocol to derive the second message $s_2$ as $H(m\concat \verkey \concat s_1)$.

Performing the Fiat-Shamir signature in superposition yields:
\[
    \ket{\sigma_m} \coloneqq \sum_{a\in A\backslash\{0\}} \ket{a} \otimes \ket{s_1^a, s_2^a = H(a\concat m\concat \verkey \concat s_1), s_3^a}
\]
where $(s_1^a, s_3^a)$ are the prover's messages in the sigma protocol using randomness $H(k\concat a)$. Verification and deletion are defined similarly to the signature construction above.

Also similarly to the previous construction, if the certificate is valid, then the verifier must have ``forgotten'' every element of $A\backslash\{0\}$. 
Thus, even if it knew some transcript $(s_1^a, s_2^a, s_3^a)$ of the sigma protocol, it could not prove to a third party that $s_2^a$ was really derived using $H(a\concat m\concat \verkey \concat s_1)$.\footnote{It is tempting to try to use a single first message $s_1$ and derive the second and third messages of the sigma protocol in superposition using $a$. However, this would lead to answering multiple challenges using the same first message, which may not be possible with a simulated $s_1$.}
Instead, the third party would suspect that $s_2^a$ was chosen carefully to match a faulty $s_1$.

Crucially, by locally programming $H$ at points that include some $a\in A$, the simulator can simulate every Fiat-Shamir transcript $(s_1^a, s_2^a, s_3^a)$ \emph{only using knowledge of $\verkey$} (or of the statement $x$). Thus, in the signature case, it no longer needs to receive anything signed under $\verkey$ to do its job.

\paragraph{NIZKs.} A useful consequence of using Fiat-Shamir to construct signatures is that the Fiat-Shamir transform can also be used for turning a sigma protocol into a NIZK. The construction is similar to the signature case, except that $s_2^a = H(a\concat m\concat \verkey \concat s_1^a)$ is replaced with $H(a\concat x \concat s_1^a)$, where $x$ is the NP statement being proven.

\paragraph{Other Technical Challenges.} 
We briefly mention two additional technical challenges that appear in our construction. First, it is not immediately obvious that Fiat-Shamir can be simulated in superposition, even if it is post-quantum secure. Previous works have addressed such issues by using collapsing protocols~\cite{EC:Unruh16} or small-range distributions~\cite{FOCS:Zhandry12}. However, both of these techniques require collapsing the argument/signature to a large degree, which is at odds with deletion: if a superposition state is indistinguishable from a a measured state, we are almost back to the classical case and there is no way to delete! To avoid this issue, we use complexity leveraging to switch each of the superposed transcripts to be simulated, one at a time.

Second, it is not immediately obvious that Fiat-Shamir is sound in a structured superposition as above. The soundness of Fiat-Shamir in the quantum setting is a highly nontrivial task, but has been shown under certain conditions in the case where the resulting argument is classical~\cite{C:DFMS19,C:LiuZha19}. To argue soundness of Fiat-Shamir in superposition, we actually use coset states. We show that if the coset offset is not known, then the coset state appears to have been measured in the computational basis.
Thus, we can treat our construction as having a classical argument/signature when proving soundness.

\subsection{Black-Box Barriers to Plain Model Constructions.}

Finally, we give an overview of the black-box barrier for the plain model. At a high level, we consider an adversary $\adv$ and a distinguisher $D$ which, as part of their auxiliary input, share a program that includes a key pair $(\widetilde{sk}, \widetilde{\verkey})$ for an internal (post-quantum) signature scheme. 
The program on some input an alleged signature $\ket{\sigma}$ for a message $m$ under key $\verkey$, verifies it and if the check passes, it signs $m\concat \verkey$ using the internal signature key $\widetilde{\sk}$. 
This operation to create a proof that $m$ was signed is gentle because of the correctness of the candidate scheme, so $\adv$ can still generate a valid certificate using the honest deletion algorithm, after it has obtained a proof for the distinguisher that $m$ was signed. Now it has produced both a valid certificate and a post-quantum signature on $m\concat \verkey$. The distinguisher can simply check the latter using $\widetilde{\verkey}$.

Observe that in the real world, $D$ will almost always obtain a valid signature on $m\concat \verkey$. On the other hand, the simulator cannot hope to extract such a signature using black-box access to $\adv$, unless it is able to forge signatures for the candidate scheme. 

On a technical level, the analysis requires generalizing a technique introduced by \cite{bbbv97} for analyzing the behavior of an oracle algorithm with a reprogrammed \emph{classical} oracle to handle oracles that do quantum computation instead. 
Roughly, we show that if an oracle algorithm is able to distinguish between oracle access to two unitaries $U_0$ and $U_1$, then outputting its query register at a random timestep produces a mixed state with noticeable probability mass on pure states $\ket{\psi}$ where $U_0\ket{\psi}$ and $U_1\ket{\psi}$ are proportionally far in trace distance.
The generalized technique may be of independent interest.


\subsection{Related Works}

\paragraph{Quantum Deniability.} \cite{STOC:ColGolVaz22} revisit the problem of deniable encryption in the quantum setting. In classical deniability, the encryptor should be able to produce ``fake'' proof to the adversarial coercer how a given ciphertext is actually an encryption of some other message. Coledangelo, Goldwasser, and Vazirani propose a uniquely quantum spin on the task: by computing the ciphertext, any explanation for it is destroyed.
Although this has similarities to our setting, their result is quite different. They ensure that the third-party coercer never sees the explanation. In contrast, the adversarial verifier necessarily sees the ``explanation'' for signatures/NIZKs - the signature/NIZK itself - but later is forced to ``forget'' it.

\paragraph{Certified Deletion.} Certified deletion was first proposed by Broadbent and Islam for encryption~\cite{TCC:BroIsl20}. It has since been generalized to a variety of other primitives, e.g.~\cite{C:HMNY22,EC:HKMNPY24,ITCS:Poremba23,TCC:AnaPorVai23,EC:AKNYY23,C:BarKhu23,EC:BGKMRR24,C:BR24,TQC:MPY24}. To the best of our knowledge, the only two works to have considered any notion of simulation in defining certified deletion are \cite{C:HMNY22} and \cite{EC:BGKMRR24}. \cite{C:HMNY22} considers certified everlasting zero knowledge, which uses a simulation definition as a result of standard definitions for zero knowledge. \cite{EC:BGKMRR24} consider a simulation-style definition of obfuscation with certified deletion in the structured oracle model as a side result, inspired by definitions of ideal obfuscation. In contrast, signatures are not typically considered to be a ``simulation primitive'', and NIZKs require either the random oracle or CRS models, which require special care for deniability.

\paragraph{Unclonability.} Unclonability prevents an adversary from transforming an object (such as a program) into two functioning copies of the object~\cite{Aar09}. It is closely related to certified deletion, since a functional copy of the object can be considered as the ``certificate''. Previously, Goyal, Malavolta, and Raizes~\cite{TCC:GMR24} considered a related notion to certified deniability under the name of ``strongly unclonable proofs''. In a strongly unclonable proof, an adversarial man-in-the-middle (MiM) who receives a simulated proof of some (potentially false) statement $x$, then interacts with two sound verifiers to prove statements $\widetilde{x}_1$ and $\widetilde{x}_2$. Strong unclonability guarantees that no MiM can convince \emph{both} verifiers of false statements.
GMR showed that in general, strongly unclonable NIZKs do not exist. Fortunately, their techniques do not extend to certified deniability. GMR's impossibility relies on an interactive verification, during which the MiM can forward messages between the two verifiers. In certified deniability, the NIZK is deleted before any messages reach the second verifier. In general, our definition is also more robust; for example, it rules out the possibility of the third-party verifier who accepts false statements with probability $1/2$ from being convinced by a deleted NIZK with even probability $1/2 + \epsilon$.
    \ifsubmission\else
        \section{Preliminaries}

All assumptions mentioned are post-quantum unless otherwise specified.
We write $\mathsf{Func}(\cX, \cY)$ to be the set of all functions $f: \cX \rightarrow \cY$. We say that two distributions $A$ and $B$ are computationally indistinguishable, denoted as $A \approx_c B$, if no QPT distinguisher $D$ with poly-size auxiliary input register $\regi{D}$ can distinguish between them with noticeable advantage. When we wish for a finer-grained treatment, we say that $A$ and $B$ are $\epsilon$-computationally indistinguishable, denoted as $A\approx_c^\epsilon B$, if the distinguishing advantage is instead bounded by $\epsilon$. In the case where $A$ or $B$ depends on a poly-sized auxiliary input register $\regi{A}$, the contents of $\regi{A}$ and $\regi{D}$ may be entangled.

\subsection{Quantum Computing}

\paragraph{Subspace States.} For any subspace $A\subset \{0,1\}^\secpar$, the subspace state $\ket{A}$ for $A$ is
\[
    \ket{A} :\propto \sum_{a\in A} \ket{a}
\]

Aaronson and Christiano~\cite{AC12} show that the projector onto $\ket{A}$ can be implemented using queries to membership oracles $O_A$ and $O_{A^\perp}$ that decide the query's membership in $A$ (respectively, $A^\perp$).

\begin{lemma}[\cite{AC12}]\label{lem:subspace-proj}
    Let $A\subset \{0,1\}^\secpar$ be a subspace. Let $P_{A} = \sum_{a\in A}\ketbra{a}$ and $P_{A^\perp} = \sum_{a\in A^\perp}\ketbra{a}$ be projectors onto the space spanned by elements of $A$ and $A^\perp$, respectively. Then
    \[
        H^{\otimes n} \mathbb{P}_{A^\perp} H^{\otimes n} \mathbb{P}_{A} = \ketbra{A}
    \]
\end{lemma}

Ben-David and Sattath show that given a random subspace state of appropriate dimension, it is hard to find a vector $v_1 \in A\backslash \{0\}$ together with a vector $v_2 \in A^\perp\backslash\{0\}$, even when given oracle access $O_A$ and $O_{A^\perp}$.

\begin{lemma}[\cite{Quantum:BS23}]\label{lem:subspace-dph}
    Let $A\subset \{0,1\}^\secpar$ be a random subspace of dimension $\secpar/2$ and let $\epsilon > 0$ be such that $1/\epsilon = o(2^{\secpar/2})$. Given one copy of $\ket{A}$ and oracle access to $O_A$ and $O_{A^\perp}$, any adversary who produces $v_1 \in A\backslash\{0\}$ together with $v_2 \in A^\perp\backslash\{0\}$ with probability $\epsilon$ requires $\Omega(\sqrt{\epsilon}2^{\secpar/2})$ queries.
\end{lemma}

\paragraph{Quantum Oracles.}
We recall a result from \cite{bbbv97} that aids in reasoning about reprogrammed oracles. Consider a quantum adversary who has quantum query access to one of two classical oracles $H$ and $H'$. 
They bound the ability of the adversary to distinguish between the two oracles in terms of the amplitude with which it queries on (classical) inputs $x$ where $H(x) \neq H'(x)$. As a simple corollary, if the adversary is able to distinguish the two oracles in a polynomial number of queries, then measuring one of its queries at random produces an $x$ such that $H_0(x) \neq H_1(x)$ with noticeable probability. 

\begin{lemma}[\cite{bbbv97}, Paraphrased]\label{lem:QROM-replacement}
    Let $H$ and $H'$ be oracles which differ on some set of inputs $\cX$. Let $\ket{\psi_i} = \sum_{y} \alpha_{y, i} \ket{\phi_{y,i}} \otimes \ket{y}_{Q}$ be the state of $A^H$ at time $i$, where $Q$ is the query register. Let $\ket{\psi_i'}$ be the state of $A^{H'}$ at time $i$.
    Then for all $T \in \mathbb{N}$,
    \[
        \tracedist(\ket{\psi_T}, \ket{\psi_T'}) \leq \sqrt{T \sum_{i=1}^{T} \sum_{x^*\in \cX} |\alpha_{x^*, i}|^2}
    \]
\end{lemma}

\paragraph{Quantum Random Oracles.} In the quantum random oracle model, all parties have access to an oracle $H \gets \Func(\cX, \cY)$ implementing a random function. To ease notation, we implicitly pad inputs to the random oracle: given $x\in \cX_1$, where $\cX = \cX_1\times\cX_2$, we denote $H(x\concat \vec{0})$ as $H(x)$.

It will be useful to sometimes be able to derive large amounts of randomness from the oracle as if it were a PRF. The following lemma formalizes this treatment.

\begin{lemma}\label{lem:qrom-prf}
    Let $H:\cX_1 \times \cX_2 \rightarrow \cY$ and $G:\cX_2 \rightarrow \cY$ be random oracles. Define $H_k: \cX_2\rightarrow \cY$ by $H_k(v) \coloneqq H(k\concat v)$. If $1/|\cX_1| = \negl(\secpar)$, then
    \[
        \left\{(O_H, O_{H_k}): \begin{array}{cc}
             H\gets \mathsf{Func}(\cX_1 \times \cX_2, \cY)  \\
             k\gets \cX_1
        \end{array}\right\}
        \approx_c
        \left\{(O_H, O_{G}) : \begin{array}{cc}
             H\gets \mathsf{Func}(\cX_1 \times \cX_2, \cY)  \\
             G\gets \mathsf{Func}(\cX_2, \cY)
        \end{array}\right\}
    \]
    where $O_{f}$ denotes oracle access to a function $f$.
\end{lemma}
\begin{proof}
    For any $k\in \cX_1$, define $H_{k,G}$ as
    \[
        H_{k,G}(x) = 
        \begin{cases}
            G(x') &\text{if } x = k\concat x' \text{ for some } x'
            \\
            H(x) &\text{else}
        \end{cases}
    \]
    The right distribution is identically distributed to
    \[
        \left\{(O_{H_{k,G}}, O_{H_k}) : \begin{array}{cc}
             H\gets \mathsf{Func}(\cX_1 \times \cX_2, \cY)  \\
             G\gets \mathsf{Func}(\cX_2, \cY)
        \end{array}\right\}
    \]
    By \cref{lem:QROM-replacement}, any distinguisher who distinguishes $(O_{H_{k,G}}, O_{H_k})$ from $(O_{H}, O_{H_k})$ with advantage $\epsilon$ in $q$ queries can produce some $x^*$ such that $O_{H_{k,G}}(x^*) \neq O_{H}(x^*)$ with probability $\epsilon^2/q^2$. Whenever this occurs, $x^* = k\concat x$ for some $x$. Since $k$ is drawn uniformly at random from $\cX_1$, it must be the case that $\epsilon^2/q^2 \leq 1/|\cX_1| = \negl(\secpar)$. If $q$ is $\poly(\secpar)$, then $\epsilon$ must be $\negl(\secpar)$.
\end{proof}

\paragraph{Z-Twirl.} It is well-known that adding a random phase to a state is equivalent to measuring the state in the computational basis. Here we present a slightly generalized form of this.
\begin{lemma}\label{lem:general-ztwirl}
    Let $\ket{\psi} = \sum_{x\in \cX} \alpha_x \ket{x} \otimes \ket{\phi_x}$ be a quantum state where $\cX$ is a vector space. Denote $\ket{\psi_s}\coloneqq \sum_{x\in \cX} \alpha_x (-1)^{s\cdot x}\ket{x} \otimes \ket{\phi_x}$ for any $s\in \cX$. Then
    \[
        \sum_{s\in \cX} \ketbra{\psi_s} = \sum_{x\in \cX} |\alpha_x|^2 \ketbra{x} \otimes \ketbra{\phi_x} 
    \]
\end{lemma}
\begin{proof}
    We compute
    \begin{align*}
        \sum_{s\in \cX^n} \ketbra{\psi_{s}}
        &=  \sum_{s\in \{0,1\}^n} \alpha_{x_1} \overline{\alpha_{x_2}}\sum_{x_1, x_2 \in \cX} (-1)^{(x_1 - x_2)\cdot s} \ket{x_1}\bra{x_2}\otimes \ket{\phi_{x_1}}\bra{\phi_{x_2}}
        \\
        &= \sum_{x_1, x_2 \in \cX} \alpha_{x_1} \overline{\alpha_{x_2}} \ket{x_1}\bra{x_2}\otimes \ket{\phi_{x_1}}\bra{\phi_{x_2}} \sum_{s\in \{0,1\}^n} (-1)^{(x_1 - x_2)\cdot s} 
        \\
        &=\sum_{x \in \cX} |\alpha_{x}|^2 \ketbra{x}\otimes \ketbra{\phi_a}
    \end{align*}
\end{proof}

\subsection{Argument Systems}\label{sec:prelims-args}

\begin{definition}[Argument]
    An argument system $(P, V)$ for an NP language $\lang$ is a (potentially interactive) protocol between a prover $P$ and a verifier $V$ where $P$ inputs a statement and a witness and $V$ outputs accept or reject. It should satisfy two properties:
    \begin{itemize}
        \item \textbf{Correctness.} If $w$ is a witness for $x\in \lang$, then at the end of the execution, $V$ outputs accept.
        \item \textbf{Soundness.} For every adversarial prover $P^*$, if $x\notin \lang$, then at the end of an execution $\langle P^*, V\rangle$, $V$ outputs reject with probability $1-\negl$. If this holds for QPT $P^*$, we call the soundness computational. If it holds for unbounded $P^*$, we call the soundness statistical.
    \end{itemize}
\end{definition}

\begin{definition}[Zero-Knowledge]
    An argument system $(P, V)$ for an NP language $\lang$ is zero knowledge if there exists a QPT algorithm $\Sim$ such that for all QPT adversaries $V^*$ with auxiliary input register $\regi{V^*}$ and all statement/witness pairs $(x,w)$ where $x\in \lang$,
    \[
        \{\langle P(x,w), V^*(\regi{V^*}\rangle\} \approx_c \{\Sim(V^*, \regi{V^*}\}
    \]
\end{definition}

A frequently useful class of argument arises from sigma protocols. Sigma protocols are three round, public-coin arguments with relaxed zero-knowledge properties. In a public-coin protocol, the verifier's messages are truly random.

\begin{definition}[Sigma Protocol]
    A Sigma protocol for an NP language $\lang$ with relation $\mathsf{Rel}_{\lang}$ is a three-message public-coin argument system $\Sigma = (P_1, P_3, \Verify_\Sigma)$ for $\lang$ with the following properties:
    \begin{itemize}
        \item \textbf{Special Soundness.} There exists an extractor $E$ which, given two transcripts $(s_1, s_2, s_3)$ and $(s_1, s_2', s_3')$ for the same statement $x$ with the same first message $s_1$ and different verifier challenges $s_2\neq s_2'$, extracts a witness for $x$.
        \item \textbf{Special Honest-Verifier Zero Knowledge (HVZK).} There exists a QPT simulator $\Sim_\Sigma$ such that for every $(x,w) \in \mathsf{Rel}_{\lang}$ and every second message $s_2$, 
        \[
            \{ (s_1, s_2, s_3): s_1 \gets P_1(x, w), s_3 \gets P_3(x, w, s_2)\} \approx_c \{\Sim_\Sigma(x, s_2)\}
        \]
        We say $\Sigma$ is $\epsilon$-secure if no QPT distinguisher has greater than $\epsilon$ advantage in distinguishing these two distributions.
    \end{itemize}
\end{definition}

\paragraph{Fiat-Shamir Transform.} 
The Fiat-Shamir transform~\cite{C:FiaSha86,C:BelGol92} modifies a sigma protocol to become non-interactive by using a random oracle $H$. 
Specifically, the prover first computes $s_1 \gets P_1(x,w)$, then computes the verifier's challenge $s_2 \gets H(x\concat s_1)$, and finally computes $s_3 \gets P_3(x,w,s_2)$ and outputs $(s_1, s_2, s_3)$. 
To verify a transcript, the verifier checks that $s_2 = H(s_1)$, then verifies the transcript using $\Verify_{\Sigma}$.
Given a sigma protocol $\Sigma$, we denote the Fiat-Shamir transform of $\Sigma$ as $\FS_\Sigma^H(x,w)$. We extend the notation as $\FS_{\Sigma}^H(x,w;r)$ when specifying the prover's randomness $r$ for $\Sigma$.

Although the analysis of Fiat-Shamir is more complicated in the quantum setting, a series of works have shown that, under mild assumptions on the sigma protocol, post-quantum Fiat-Shamir is sound even with quantum access to $H$, culminating in~\cite{C:DFMS19,C:LiuZha19}. 

\begin{theorem}[\cite{C:LiuZha19}]\label{thm:fs-pq-nizk}
    If a post-quantum sigma protocol has (1) perfect completeness, (2) quantum proof of knowledge, and (3) unpredictable first messages, then the Fiat-Shamir heuristic gives a quantum NIZKPoK.
\end{theorem}

Fiat-Shamir can also be used to create a signature scheme from a sigma protocol by using an NP instance $x$ where finding the witness is hard (e.g. the image of a one-way function) as the public key and its witness $w$ as the secret key. A signature on a message $m$ is obtained by computing $\FS_\Sigma^{H(m\concat \cdot)}(x,w)$, where $H(m\concat \cdot)$ denotes $H$ with the first portion of the input fixed to $m$.

\cite{C:DFMS19,C:LiuZha19} also showed that if the underlying sigma protocol is collapsing,\footnote{\cite{C:DFMS19} refers to this property as ``computational unique responses''.} then the Fiat-Shamir transform gives a secure signature scheme in the quantum random oracle model. 

\begin{definition}[Collapsing for Sigma Protocols~\cite{AC:Unruh16,C:DFMS19,C:LiuZha19,FOCS:CMSZ21}]
    We say a protocol $\langle P, V\rangle$ is \emph{collapsing} if for every polynomial-size interactive quantum adversary $P^*$ and polynomial-size quantum distinguisher $\adv$,
    \[
        \left| \Pr[1\gets \mathsf{CollapseExp}(0, P^*, \adv)] 
        - \Pr[1\gets \mathsf{CollapseExp}(1, P^*, \adv)]\right| \leq \negl(\secpar)
    \]
    For $b\in \{0,1\}$, the experiment $\mathsf{CollapseExp}(b, P^*, \adv)$ is defined as follows:
    \begin{enumerate}
        \item The challenger executes $\langle P^*, V\rangle$, storing the result in registers $(\regi{1}, \dots, \regi{n})$. It measures every register $\regi{n}$ in the computational basis.
        \item The challenger coherently evaluates $V(\ket{m_1, \dots, m_n})$ and measures the result. If it is reject, the experiment aborts by outputting a random bit.
        \item If $b=0$, the challenger does nothing. If $b=1$, the challenger measures $\regi{n}$ in the computational basis.
        \item The challenger sends $(\regi{1}, \dots, \regi{n})$ to $\adv$. The experiment outputs $\adv$'s output bit.
    \end{enumerate}
\end{definition}

\begin{theorem}[\cite{C:LiuZha19}]\label{thm:fs-pq-sig}
    If a post-quantum sigma protocol is collapsing, then the Fiat-Shamir heuristic gives a secure post-quantum digital signature scheme in the quantum random oracle model.
\end{theorem}

We note that natural sigma protocols satisfying the requirements of both \cref{thm:fs-pq-nizk,thm:fs-pq-sig} are known. For instance, Unruh~\cite{EC:Unruh12} shows that a slight modification of Blum's Hamiltonian path argument is a quantum proof of knowledge. It also has perfect completeness and unpredictable first messages. Furthermore, if the commitments in the first round are implemented via a random oracle, then it has subexponential HVZK.

\subsection{Revocable Signatures and NIZKs}

\begin{definition}
    A digital signature is a tuple of algorithms $(\Gen, \Sign, \Verify)$ with the following behavior:
    \begin{itemize}
        \item $\underline{\Gen(1^\secpar)}$ takes in the security parameter and outputs a key pair $(\sk, \verkey)$.
        \item $\underline{\Sign(\sk, m)}$ takes in a signing key $\sk$ and a message $m$, then outputs a (potentially quantum) signature $\sigma$.
        \item $\underline{\Verify(\verkey, \sigma, m)}$ takes in a verification key $\verkey$, an alleged signature $\sigma$, and a message $m$, then outputs accept or reject.
    \end{itemize}
    A digital signature must statisfy the following properties:
    \begin{itemize}
        \item \textbf{Correctness.} For every message $m$:
        \[
            \Pr\left[\Accept \gets \Verify(\verkey, \sigma, m):
            \begin{array}{c}
                 (\sk, \verkey)\gets \Gen(1^\secpar)  \\
                 \sigma \gets \Sign(\sk, m)
            \end{array}\right] \geq 1-\negl(\secpar)
        \]
        \item \textbf{Existential Unforgeability under Chosen Message Attack (EUF-CMA).} For all QPT adversaries $\adv$, 
        \[
            \Pr\left[
            \begin{array}{c}
                m \notin M \\
                \land \\
                \Accept \gets \Verify(\verkey, \regi{sig}, m)  
            \end{array}:
                    \begin{array}{cc}
                        (\sk, \verkey)\gets \Gen(1^\secpar) \\
                        \adv^{\Sign(\sk, \cdot)}(\verkey) 
                    \end{array}\right] 
            \leq \negl(\secpar)
        \]
        where $M$ is the list of messages that the adversary queries to the signing oracle $\Sign(\sk, \cdot)$.
    \end{itemize}
\end{definition}

\paragraph{Revocable Signatures.} 

A revocable signature scheme~\cite{TQC:MPY24} augments the signature scheme syntax with two additional algorithms $\Del$ and $\DelVer$. Additionally, $\Sign(\sk, m)$ is modified to output both a signature register $\regi{sig}$ and a deletion verification key $\delkey$. The new algorithms act as follows:
\begin{itemize}
    \item $\underline{\Del(\regi{sig})}$ takes in a register containing a signature, then outputs a certificate register $\regi{cert}$.
    \item $\underline{\DelVer(\delkey, \regi{cert})}$ takes in the deletion verification key and a certificate register, then outputs accept or reject.
\end{itemize}

Additionally, a revocable signature scheme should satisfy deletion correctness and a notion of revocation security. We omit \cite{TQC:MPY24}'s notion of revocation security here. Instead, we define certified deniability for (revocable) signatures in \cref{sec:sig-cden-def}.

\begin{definition}[Deletion Correctness]
    For all messages $m$, 
    \[
        \Pr\left[ \Accept \gets \DelVer(\delkey, \regi{cert}):
                \begin{array}{c}
                     (\sk, \verkey) \gets \Gen(1^\secpar)  \\
                     (\regi{sig}, \delkey) \gets \Sign(\sk, m) \\
                     \regi{cert} \gets \Del(\regi{sig})
                \end{array}\right]
        \geq 1-\negl(\secpar)
    \]
\end{definition}

\paragraph{Revocable NIZKs.} Although these have not be explicitly defined before, they follow similar syntax to revocable signatures, so we include the description in this section. A revocable NIZK is augmented with two additional algorithms $\Del$ and $\DelVer$, which act similarly to their signature counterparts. Additionally, a revocable NIZK must satisfy deletion correctness. We define certified deniability for (revocable) NIZKS in \cref{sec:nizk-cden-def}.


    \fi
    \section{Definitions of Certified Deniability}\label{sec:cden-defs}
\subsection{Signatures}\label{sec:sig-cden-def}
\justin{Our construction actually satisfies the NIZK notion, which should imply this (and be strictly stronger).}

Deniable authentication was initially defined by Dwork, Naor, and Sahai~\cite{STOC:DwoNaoSah98} using the simulation paradigm. Informally, a signature is deniable if it could be \emph{simulated} by using only public information. 

We follow a similar simulation-based paradigm which uses a real experiment $\sigcd_{\adv(\regi{\adv})}(\verkey)$. This experiment is parameterized by a QPT adversary $\adv$ who receives auxiliary input in register $\regi{\adv}$ and a key pair $(\sk, \verkey)$. It is defined as follows.
\begin{enumerate}
    

    \item $\adv$ is initialized with register $\regi{\adv}$ and $\verkey$. The challenger is initialized with $\sk$.
    \item $\adv$ gets access to a signing oracle for $\sk$. The oracle tracks the set of messages $M$ that it has been queried on. 
    \item At any point, $\adv$ may make a single special query for a signature on some $m^* \notin M$. The signing oracle responds with $(\ket{\sigma}, \delkey) \gets \Sign(\sk,m)$, but does not add $m^*$ to $M$. Afterwards, $\adv$ may not query again on $m^*$.
    \item $\adv$ outputs a certificate register $\regi{cert}$ and a register $\regi{\adv}$. 
    \item If $\DelVer(\delkey, \regi{cert}) = \Accept$, then output $(M, \regi{\adv})$. Otherwise, output $(M, \bot)$.
\end{enumerate}
If working in an oracle model where the parties have access to an oracle $H$, then we denote the experiment as $\sigcd^H_{\adv(\regi{\adv})}$. 

\begin{definition}[Certified Deniability for Signatures: Plain Model]\label{def:cden-sigs}
    A revocable signature scheme $(\Gen,\allowbreak \Sign,\allowbreak \Verify,\allowbreak \Del,\allowbreak \DelVer)$ is \emph{certifiably deniable} if for every QPT adversary $\adv$, there exists a QPT simulator $\Sim$ such that for every QPT adversary $\adv$ with poly-size auxiliary input register $\regi{\adv}$,
    \begin{gather*}
       \left\{(\verkey, \sigcd_{\adv(\regi{\adv})}(\sk, \verkey)) : (\sk, \verkey) \gets \Gen(1^\secpar) \right\}
        \\
        \approx_c
        \\
        \left\{(\verkey, M_\Sim, \Sim^{\Sign(\sk, \cdot)}(\adv, \regi{\adv}, \verkey)) : (\sk, \verkey) \gets \Gen(1^\secpar) \right\}
    \end{gather*}
    where $M_\Sim$ is the set of messages on which $\Sim$ queries $\Sign(\sk, \cdot)$.
\end{definition}

The set $M$ acts as a way to restrict $\Sim$ from querying the signing oracle on $m^*$, which is only announced after $\adv$ sees $\verkey$. $\adv$ gets a free signing query on $m^*$ that is not recorded, but all of $\Sim$'s queries are recorded.

The above definition requires that $\adv$ declare whether it will delete a signature on $m^*$ when it queries for a signature on $m^*$. A more general notion might allow the adversary to see many signatures before deciding which one to delete. Our definition can be upgraded to satisfy the more general notion using standard complexity leveraging techniques.

\paragraph{Certified Deniability in the QROM.} We may also define certified deniability in the CRS model or the QROM model. Following~\cite{C:Pass03}'s definition from the classical setting, a deniable simulator does \emph{not} have the ability to program the global random oracle; this is enforced by sampling a fresh random oracle before the experiment and including its description in the output of the experiment. Thus, any simulator which attempts to pretend that the oracle had different behavior is will be caught.

\begin{definition}[Certified Deniability for Signatures: QROM]\label{def:cden-sig-qrom}
    A revocable signature scheme $(\Gen,\allowbreak \Sign,\allowbreak \Verify,\allowbreak \Del,\allowbreak \DelVer)$ is \emph{certifiably deniable in the quantum random oracle model} if there exists a QPT simulator $\Sim$ such that for every QPT adversary $\adv$ with poly-size auxiliary input register $\regi{\adv}$,
    \begin{gather*}
        \left\{(O_H, \verkey, \sigcd^H_{\adv(\regi{\adv})}(\sk, \verkey)) : \begin{array}{c}
             H\gets \mathsf{Func}(\cX, \cY)  \\
            (\sk, \verkey) \gets \Gen^H(1^\secpar)
        \end{array} \right\}
        \\
        \approx_c
        \\
        \left\{(O_H, \verkey, M_\Sim, \Sim^{H, \Sign^H(\sk, \cdot)}(\adv, \regi{\adv}, \verkey)) : \begin{array}{c}
             H\gets \mathsf{Func}(\cX, \cY)  \\
            (\sk, \verkey) \gets \Gen^H(1^\secpar)
        \end{array} \right\}
    \end{gather*}
    where $O_H$ denotes oracle access to $H$ and $M_\Sim$ is the set of messages on which $\Sim$ queries $\Sign^H(\sk, \cdot)$.
\end{definition}

We show in \cref{sec:fs-cden} how to obtain signatures with certified deniability by adding certified deniability to the Fiat-Shamir transformation.

\subsection{NIZKs}\label{sec:nizk-cden-def}

\paragraph{Certified Deniability: Real Experiment.} Certified deniability follows the standard real/ideal paradigm of simulator-based definitions. The real experiment $\nizkcd_{\adv(\regi{\adv})}(s, w)$ is parameterized by an adversarial QPT algorithm with auxiliary input $\regi{\adv}$ and some NP statement and witness pair $(s, w)$.  $\nizkcd_{\adv(\regi{\adv})}(s, w)$ consists of the following distribution:
\begin{enumerate}
    \item Sample $(\delkey, \ket{\pi})\gets \Prove(s, w)$ and send $\ket{\pi}$ to the adversary.
    \item The adversary outputs two registers $(\regi{\cert}, \regi{\adv})\gets \adv(\regi{\adv}, \ket{\pi})$.
    \item If $\DelVer(\delkey, \regi{\cert})$ outputs accept, then the experiment outputs the adversary's residual state register $\regi{\adv}$. Otherwise, it outputs $\bot$.
\end{enumerate}
If working in an oracle model where the parties have access to an oracle $H$, then we denote the experiment as $\nizkcd^H_{\adv(\regi{\adv})}(s, w)$.

\begin{definition}[Certified Deniability for NIZKs: Plain Model]\label{def:nizk-cden-plain}
    A non-interactive argument system $(\Prove, \Verify, \DelVer)$ is \emph{certifiably deniable} if there exists a QPT simulator $\Sim$ such that for every QPT adversary $\adv$ with poly-size auxiliary input register $\regi{\adv}$ and every NP statement/witness pair $(s,w)$,
    \[
       \left \{\nizkcd_{\adv(\regi{\adv})}(s, w) \right\} 
       \approx_c 
       \left\{\Sim(s, \adv, \regi{\adv}) \right\}
    \]
\end{definition}

We may also define certified deniability in the CRS model or the QROM model. Following~\cite{C:Pass03}'s definition from the classical setting, a deniable simulator does \emph{not} have the ability to program the global random oracle; this is enforced by sampling a fresh random oracle before the experiment and including its description in the output of the experiment. Thus, any simulator which attempts to pretend that the oracle had different behavior is likely to be caught.

\begin{definition}[Certified Deniability for NIZKs: QROM]\label{def:nizk-cden-qrom}
    A non-interactive argument system $(\Prove, \Verify, \DelVer)$ is \emph{certifiably deniable in the quantum random oracle model} if there exists a QPT simulator $\Sim$ such that for every QPT adversary $\adv$ with quantum auxiliary input $\aux$ and every NP statement/witness pair $(s,w)$,
    \begin{gather*}
        \left\{ \left(\nizkcd^H_{\adv(\regi{\adv})}(s, w), O_H \right) : H\gets \mathsf{Func}(\cX, \cY) \right\} 
        \\
        \approx_c 
        \\
        \left\{\left(\Sim^H(s, \Adv, \regi{\adv}), O_H \right) : H\gets \mathsf{Func}(\cX, \cY) \right\}
    \end{gather*}
    where $O_H$ denotes oracle access to $H$.
\end{definition}

We show in \cref{sec:fs-cden} how to obtain NIZKs with certified deniability by adding certified deniability to the Fiat-Shamir transformation.

    \section{Fiat-Shamir with Certified Deniability}\label{sec:fs-cden}

In this section, we show how to modify the Fiat-Shamir transform to add certified deniability. As a result of this general paradigm, we obtain both signatures with certified deniability (see \cref{sec:sig-cden-def}) and non-interactive zero knowledge with certified deniability (see \cref{sec:nizk-cden-def}).

Let $\Sigma$ be a Sigma protocol, and denote $\FS_{\Sigma}^H$ as its Fiat-Shamir transform with oracle access to $H$. We denote an oracle with the first part of the input fixed to $v$ as $H(v\concat \cdot)$. On query $w$, it returns $H(v\concat w)$.

\begin{construction}[$\fscd_{\Sigma}$]\label{constr:fs-cden}
    The construction is as follows.
    \begin{itemize}
        \item \underline{$\Prove(x, w)$:}
        \begin{enumerate}
            \item Sample a subspace $A\subset \{0,1\}^\secpar$ of dimension $\secpar/2$ and an offset $s\gets \{0,1\}^\secpar$, then prepare the coset state $\ket{A_{0,s}} \propto \sum_{a\in A}(-1)^{a\cdot s}\ket{a}$ in register $\regi{A}$. Initialize register $\regi{\Sigma} = (\regi{\Sigma,1}, \regi{\Sigma,2}, \regi{\Sigma,3})$ to $\ket{0}$.

            \item Sample a key $k\gets \{0,1\}^\secpar$ and apply the isometry 
            \[
                \ket{a}_{\regi{A}} \otimes \ket{\vec{0}}_{\regi{\Sigma}}
                \mapsto
                \ket{a}_{\regi{A}} \otimes \ket{\FS^{H(a\concat\cdot)}_{\Sigma}(x,w; H(k\concat a))}_{\regi{\Sigma}}
            \]
            to registers $\regi{A}$ and $\regi{\Sigma}$. This results in a state
            \[
                \ket{\pi} 
                :\propto 
                \sum_{a\in A} (-1)^{s\cdot a} \ket{a}_{\regi{A}} \otimes \ket{s_1^a, s_2^a, s_3^a}_{\regi{\Sigma}}
            \]
            where $s_2^a = H(a\concat s_1^a)$ and $(s_1^a, s_3^a)$ are the prover's messages from $\Sigma$ using randomness $H(k\concat a)$.
            \item Output $\ket{\pi}$ as the argument and $\delkey \coloneqq (A, k, w)$ as the deletion key.

        \end{enumerate}
        
        \item \underline{$\Verify(x, \regi{arg})$:} 
        \begin{enumerate}
            \item Parse $\regi{arg} = (\regi{A}, \regi{\Sigma})$.
            \item Coherently evaluate $\FS_\Sigma.\Verify$ on register $\regi{\Sigma}$, then measure and output the result.
        \end{enumerate}
        
        \item \underline{$\Del(\regi{arg})$:} Output $\regi{arg}$.
        
        \item \underline{$\DelVer(\delkey, \regi{cert})$:} 
        \begin{enumerate}
            \item Parse $\delkey = (A, k, w)$. Parse $\regi{arg} = (\regi{A}, \regi{\Sigma})$.

            \item Apply the isometry 
             \[
                \ket{a}_{\regi{A}} \otimes \ket{y}_{\regi{\Sigma}}
                \mapsto
                \ket{a}_{\regi{A}} \otimes \ket{y \oplus \FS^{H(a\concat \cdot)}_{\Sigma}(x,w; H(k\concat a))}_{\regi{\Sigma}}
            \]

            \item Measure register $\regi{A}$ with respect to the PVM $\{\ketbra{A_{0,s}}, I - \ketbra{A}_{0,s}\}$. Output accept if the result is the former, and reject if it is the latter.
        \end{enumerate}
    \end{itemize}

\end{construction}

\begin{theorem}\label{thm:delnizkexists}
    If $\Sigma$ is a $2^{-\secpar}$-secure sigma protocol for a language $\lang$ with (1) perfect completeness and (2) quantum proof of knowledge, then \Cref{constr:fs-cden} is a NIZKPoK for $\lang$ with certified deniability in the QROM.
\end{theorem}
\paragraph{Remark:} As mentioned in \cref{sec:prelims-args}, such sigma protocols exist unconditionally in the QROM.
\begin{proof}
    We prove certified deniability in \cref{claim:NIZK-cden}. Correctness of deletion follows from inspection and the fact that $\ket{A_{0,s}\backslash \{0\}}$ is negligibly close in trace distance to $\ket{\ket{A_{0,s}}}$.

    By \cref{lem:general-ztwirl}, any NIZK $\ket{\pi}$ is perfectly identical, from the verifier's point of view, from whether the first register $\regi{A}$ has been measured in the computational basis. Thus, \cref{constr:fs-cden} is sound and zero-knowledge if sampling a random $a$ and outputting $a \concat \FS^{H(a\concat \cdot)}(x,w)$ is sound and zero-knowledge. Observe that $a$ can be considered to be part of the first message of underlying sigma protocol, which makes the sigma protocol have unpredictable first messages. Then \cref{thm:fs-pq-nizk} implies that $a \concat \FS^{H(a\concat \cdot)}(x,w)$ is sound and zero-knowledge.
\ifsubmission\qed\else\fi\end{proof}

\cref{constr:fs-cden} is already a NIZK, if instantiated with an appropriate sigma protocol. To obtain a signature scheme, a few more details are needed.

\begin{construction}[Signature with Certified Deniability]\label{constr:sig-sim}
    Let $f:\{0,1\}^\secpar \mapsto \{0,1\}^{\poly(\secpar)}$ be a one-way function and let $\Sigma$ be a sigma protocol for the language \[
    \lang_f = \left\{x\in \{0,1\}^{\poly(\secpar)}: \exists w\in \{0,1\}^\secpar \text{ s.t. } f(w) = x\right\}
    \]
    \begin{itemize}
        \item $\underline{\Gen(1^\secpar):}$ Sample $w \gets \{0,1\}^\secpar$. Output $\verkey = f(w)$ and $\sk = w$.
        \item $\underline{\Sign(\sk, m):}$ Evaluate $(\ket{\pi}, \delkey) \gets \fscd_{\Sigma}.\Prove^{H(m\concat \cdot)}(\verkey, \sk)$ and output the result.
        \item $\underline{\Verify(\verkey, m, \regi{arg}):}$ Evaluate $\fscd_{\Sigma}.\Verify^{H(m\concat \cdot)}(\verkey, \regi{arg})$.
        \item $\underline{\Del(\regi{arg}):}$ Output $\fscd_{\Sigma}.\Del^{H(m\concat \cdot)}(\regi{arg})$.
        \item \underline{$\DelVer(\delkey, \regi{cert})$:} Output $\fscd_{\Sigma}.\DelVer^{H(m\concat \cdot)}(\delkey, \regi{cert})$.
    \end{itemize}
\end{construction}

\begin{theorem}\label{thm:delsigexists}
    If $f$ is a one-way function and $\Sigma$ is a collapsing, $2^{-\secpar}$-secure sigma protocol for $\lang_f$, then \cref{constr:sig-sim} is a signature scheme with certified deniability in the QROM.
\end{theorem}
\begin{proof}
    We prove certified deniability in \cref{claim:NIZK-cden}. Correctness of deletion follows from inspection and the fact that $\ket{A_{0,s}\backslash \{0\}}$ is negligibly close in trace distance to $\ket{\ket{A_{0,s}}}$.
    
     By \cref{lem:general-ztwirl}, any NIZK $\ket{\pi}$ is perfectly identical, from the verifier's point of view, from whether the first register $\regi{A}$ has been measured in the computational basis. Thus, \cref{constr:fs-cden} is existentially unforgeable if sampling a random $a$ and outputting $a \concat \FS_{\Sigma}^{H(m\concat a\concat \cdot)}(x,w)$ is existentially unforgeable. Observe that $a$ can be considered to be part of the first message of the underlying sigma protocol without affecting its properties. Then \cref{thm:fs-pq-sig} implies that $a \concat \FS_{\Sigma}^{H(m\concat a\concat \cdot)}(x,w)$ is existentially unforgeable.
\ifsubmission\qed\else\fi\end{proof}

\subsection{Proof of Certified Deniability}

\begin{claim}\label{claim:NIZK-cden}
      If $\mathsf{Sigma}$ is a $2^{-\secpar}$-secure sigma-protocol for a language $\lang$, then \Cref{constr:fs-cden} satisfies certified deniability for NIZKs in the QROM (\cref{def:nizk-cden-qrom}). Furthermore, assuming the same, $\Cref{constr:sig-sim}$ satisfies certified deletion for signatures in the QROM (\cref{def:cden-sig-qrom}).
\end{claim}
\begin{proof}
    The proofs of the two sub-claims are almost identical, except for a slight difference in the simulators. The NIZK simulator $\Sim$ uses statement $x$ and is given below. The signature simulator uses the verification key $\verkey$ instead of $x$, and additionally forwards any adversarial queries for signatures to the signing oracle, except if the adversary would query for $m^*$. Furthermore, in the case of signatures, we regard $H(m\concat \cdot)$ to be the random oracle, so we omit explicitly stating $H(m\concat \dots)$ below.

    On input a statement $x$, the adversary's code $\adv$ and auxiliary input register $\regi{\adv}$, the simulator $\Sim$ works as follows:
    \begin{enumerate}
        \item Sample a subspace $A\subset \{0,1\}^\secpar$ of dimension $\secpar/2$ and sample an offset $s\gets \{0,1\}^\secpar$,then prepare the coset state $\ket{A_{0,s}\backslash\{0\}} \propto \sum_{a\in A\backslash\{0\}}(-1)^{a\cdot s}\ket{a}$ in register $\regi{A}$. Initialize register $\regi{\Sigma} = (\regi{\Sigma,1}, \regi{\Sigma,2}, \regi{\Sigma,3})$ to $\ket{0}$.

        \item Apply the isometry
        \[
            \ket{a}_{\regi{A}} \otimes \ket{0}_{\regi{\Sigma}}
            \mapsto
            \ket{a}_{\regi{A}} \otimes \ket{\Sim_{\Sigma}(x, H(k_\ch \concat a); H(k\concat a)}_{\regi{\Sigma}}
        \]
        to registers $\regi{A}$ and $\regi{\Sigma}$.
        \[
            \ket{\widetilde{\pi}}
            :\propto
            \sum_{a\in A\backslash\{0\}} (-1)^{a\cdot s}\ket{a}_{\regi{A}} \otimes \ket{\widetilde{s_1^a}, \widetilde{s_2^a}, \widetilde{s_3^a}}_{\regi{\Sigma}}
        \]
        where $\widetilde{s_2^a} = H(k_\ch \concat a)$ and $(\widetilde{s_1^a}, \widetilde{s_3^a})$ are obtained by running the honest verifier zero knowledge simulator $\Sim_\Sigma$ for $\Sigma$.
        
        \item Sample a key $k_{\ch} \gets \{0,1\}^\secpar$ and define the random oracle $H'$ as follows:
        \[
            H'(q_1\concat q_2\concat q_3) \coloneqq 
            \begin{cases}
                H(k_{\ch}\concat q_2) & \text{if } q_1 \in A \text{ and } q_2\concat q_3 = x \concat \widetilde{s_1^a}
                \\
                H(q_1\concat q_2 \concat q_3) & \text{else.}
            \end{cases}
        \]
        where $\left(\widetilde{s_1^a}, \widetilde{s_2^a}, \widetilde{s_3^a}\right) = \Sim_\Sigma(x, H(k_{\ch}\concat q_2); H(k\concat a))$ for all $a\in A\backslash\{0\}$.
        Note that this may be efficiently evaluated on any computational basis queries using the description of $A$, and thus on any quantum queries in general.

        \item Compute the adversary's output $(\regi{cert}, \regi{\adv}) \gets \adv^{H'}(\regi{\adv}, \ket{\widetilde{\pi}})$.

        \item Sample a key $k\gets \{0,1\}^\secpar$. Parse $\regi{cert} = (\regi{A}, \regi{\Sigma})$ and compute the isometry 
        \[
            \ket{a}_{\regi{A}} \otimes \ket{y}_{\regi{\Sigma}}
            \mapsto
            \ket{a}_{\regi{A}} \otimes \ket{y \oplus \Sim_\Sigma(x, H(k_{\ch}\concat q_2); H(k\concat a))}_{\regi{\Sigma}}
        \]
        on registers $\regi{A}$ and $\regi{\Sigma}$.

        \item Measure register $\regi{A}$ with respect to the PVM $\{\ketbra{A_{0,s}}, I - \ketbra{A_{0,s}}\}$. If the result is the former, output $\regi{\adv}$. Otherwise output $\bot$.
    \end{enumerate}

    We now show that this simulator satisfies \cref{def:nizk-cden-qrom}, i.e. the joint distribution over the output of the simulator and the description of $H$ is indistinguishable from the real certified deniability experiment. Consider the following hybrid experiments:
    \begin{itemize}
        \item $\underline{\Hyb_0(w)}$: The real certified deniability experiment 
        $\nizkcd^H_{\adv(\regi{\adv})}(x, w)$.
        
        \item $\underline{\Hyb_1(w)}$: 
        The only difference from $\Hyb_1$ is that we replace the oracle $H$ with a reprogrammed oracle $H'_1$, defined as\footnote{Note that queries to $H'_1$ can be answered lazily, preventing an exponential blow-up from reprogramming an exponential number of positions.}
        \begin{equation}\label{eq:nizk-reprogrammed-oracle}
             H'_1(q_1\concat q_2 \concat q_3) \coloneqq 
            \begin{cases}
                H(k_{\ch}\concat q_2) & \text{if } q_1 \in A\backslash\{0\} \text{ and } q_2\concat q_3 = x \concat s_1^a
                \\
                H(q_1\concat q_2) & \text{else.}
            \end{cases}
        \end{equation}
        where $s_1^a = P_{\Sigma,1}((x,w); H(k\concat q_2))$.\footnote{We abstract the expansion of $s_1^a$ out of the definition of $H'_1$ to emphasize that $s_1^a$ is treated as an implicitly defined parameter to the reprogrammed oracle.} In the security game $\Hyb_1(w) = \nizkcd^{H'_1}_{\adv(\regi{\adv})}(x, w)$, this results in a NIZK 
        \[
            \ket{\pi} \propto \sum_{a\in A\backslash\{0\}} (-1)^{a\cdot s} \ket{a} \otimes \ket{s_1^a, \widetilde{s_2^a}, s_3^a}
        \] 
        where $\widetilde{s_2^a} = H(k_{\ch}\concat q_2)$ and $(s_1^a, s_3^a)$ are computed honestly using the witness $w$. 

        \item $\underline{\Hyb_2 = \Sim}$: The only differences from $\Hyb_1$ are for each $a\in A\backslash\{0\}$, we replace $(s_1^a, s_3^a)$ with an honest-verifier simulated transcript $\left(\widetilde{s_1^a}, \widetilde{s_3^a}\right)$ from
        \[
        \left(\widetilde{s_1^a}, \widetilde{s_2^a}, \widetilde{s_3^a}\right) = \Sim_{\Sigma}(x, H(k_\ch\concat a); H(k\concat a))
        \]
        and update the random oracle accordingly, as defined in \cref{eq:nizk-reprogrammed-oracle} with $s_1^a = \widetilde{s_1^a}$. This results in the oracle $H'$.
        
        In the security game, this modification results in a NIZK
        \[
            \ket{\widetilde{\pi}} \propto \sum_{a\in A\backslash\{0\}} (-1)^{a \cdot s}\ket{a} \otimes \ket{\widetilde{s_1^a}, \widetilde{s_2^a}, \widetilde{s_3^a}}
        \]
        Note that the definition of $\left(\widetilde{s_1^a}, \widetilde{s_3^a}\right)$ is implicit, and they are only computed coherently as necessary: (1) in answering queries to $H'_2$, (2) in computing $\ket{\widetilde{\pi}}$, and (3) in applying the isometry $\ket{a} \otimes \ket{y} \mapsto \ket{a} \otimes \ket{y \oplus \left(\widetilde{s_1^a}, \widetilde{s_2^a}, \widetilde{s_3^a}\right)}$ to check the deletion certificate.
    \end{itemize}

    We now show in \Cref{claim:nizk-hyb0-1,claim:nizk-hyb1-2,claim:nizk-hyb2-2} that
    \begin{align}
        \left(\Hyb_0, O_H \right) 
        &\approx_c \left(\Hyb_1, O_{H'_1} \right) 
        \\
        &\approx_c \left(\Hyb_2, O_{H'} \right) 
        \\
        &\approx_c \left(\Hyb_2, O_H \right) 
    \end{align}
    over the choice of the random oracle $H$.

    \begin{claim}\label{claim:nizk-hyb0-1}
        \[
        \left\{\left(\Hyb_0, O_H \right) : H\gets \mathsf{Func}(\cX, \cY)\right\}
        \approx 
        \left\{\left(\Hyb_1, O_{H'_1} \right) : H\gets \mathsf{Func}(\cX, \cY)\right\} 
        \]
    \end{claim}
    \begin{proof}
        Consider the sub-hybrid experiment $\Hyb_0'$:
        \begin{itemize}
            

            \item $\underline{\Hyb_0':}$ This hybrid is identical to $\Hyb_0$, except that $H$ is replaced with a reprogrammed oracle $H'_0$, defined as
            \[
                H'_0(v) \coloneqq 
                \begin{cases}
                    G(v') & \text{if } v = k_\ch \concat v' \text{ for some } v'
                    \\
                    H(v) & \text{else.}
                \end{cases}
            \]
            where $\cX = \{0,1\}^\secpar \times \cX_2$ and $G\gets \mathsf{Func}(\cX_2, \cY)$ is a fresh random oracle.

            \item $\underline{\Hyb_1:}$ The only difference between $\Hyb_0'$ and $\Hyb_1$ is the challenger samples a key $k_\ch\gets \{0,1\}^\secpar$ and $G$ is defined by $G(v) = H(k_\ch\concat v)$.
        \end{itemize}
        Since $G$ is a truly random function and is only used in the definition of $H_0'$,
        \[
            \left\{\left(\Hyb_0, O_H \right) : H\gets \mathsf{Func}(\cX, \cY)\right\} 
            = 
            \left\{\left(\Hyb_0', O_{H_0'} \right) : H\gets \mathsf{Func}(\cX, \cY)\right\}
        \]
        Furthermore, \cref{lem:qrom-prf} implies\footnote{Since $G$ is used to define $H_0'$, changing $G$ to match $H(k_\ch\concat \cdot)$ causes $H_0'$ to become $H_1'$.}
        \[
            \left\{\left(\Hyb_0', O_{H_0'}\right) : H\gets \mathsf{Func}(\cX, \cY)\right\} 
            \approx_c 
            \left\{\left(\Hyb_1, O_{H_1'} \right) : H\gets \mathsf{Func}(\cX, \cY)\right\}
        \]
    \ifsubmission\qed\else\fi\end{proof}

    \begin{claim}\label{claim:nizk-hyb1-2}
        If $\mathsf{Sigma}$ is a Sigma-protocol for a language $\lang$ with $2^{-\secpar}$-security, where $\nu(\secpar)$ is a negligible function, then
        \[
        \left\{\left(\Hyb_1, O_{H'_1} \right) : H\gets \mathsf{Func}(\cX, \cY)\right\}
        \approx 
        \left\{\left(\Hyb_1, O_{H'_2} \right) : H\gets \mathsf{Func}(\cX, \cY)\right\} 
        \]
    \end{claim}
    \begin{proof}
        Let $a_i$ be the $i$'th lexicographically least element of $A$. We define a series of $3\cdot|A| = 3\cdot 2^{\secpar/2}$ intermediate hybrid experiments $\Hyb_{1,i,1}$, $\Hyb_{1,i,2}$, and $\Hyb_{1,i,3}$ for $i\in [|A|]$, as well as an additional hybrid experiment $\Hyb_{1}'$. 
        \begin{itemize}
            \item $\underline{\Hyb_1'}:$ This hybrid is identical to $\Hyb_1$, except that the challenger uses an additional random oracle $G\gets \mathsf{Func}(\{0,1\}^\secpar,\cY)$. Whenever it would evaluate $H(k\concat v)$ for some $v$, it instead (coherently) evaluates $G(v)$. 

            \item $\underline{\Hyb_{1,i, 1}:}$ This identical to $\Hyb_1'$, except for the following changes. First, for every $j<i$, the honestly computed $(s_1^{a_j}, s_3^{a_j})$ are replaced by $(\widetilde{s_1^{a_j}}, \widetilde{s_3^{a_j}})$, which are generated by $\Sim_{\Sigma}$ using randomness $G(a_j)$. Since $s_1^{a_j}$ is modified, the random oracle is also updated accordingly to become $H'_{1, i}$, as defined by \cref{eq:nizk-reprogrammed-oracle} with $s_1^{a_j} \coloneqq \widetilde{s_1^{a_j}}$ for $j<i$.

            The second difference is purely syntactic. Instead of getting access to $G$, the challenger gets access to $G_{a_i}$, which is identical to $G$ except that $G_{a_i}(a_i) = \bot$. Additionally, it receives a classical copy of $(s_1^{a_i}, s_3^{a_i})$, which are computed be the honest $\Sigma$ prover using $(x,w)$ and randomness $G(a_i)$.

            \item $\underline{\Hyb_{1,i, 2}:}$ This is identical to $\Hyb_{1,i, 2}$, except that  the honestly computed $(s_1^{a_i}, s_3^{a_i})$ are replaced by $(\widetilde{s_1^{a_i}}, \widetilde{s_3^{a_i}})$, where
            \[
                \left(\widetilde{s_1^{a_i}}, \widetilde{s_2^{a_i}}, \widetilde{s_3^{a_i}}\right) = \Sim_{\Sigma}(x, H(k_\ch\concat {a_i}); G(a_i))
            \]

            \item $\underline{\Hyb_2:}$ The only difference between $\Hyb_2$ and $\Hyb_{1,i_{\max}, 2}$, where $i_{\max} = 2^{\secpar/2}$, is the challenger samples a key $k\gets \{0,1\}^\secpar$ and $G$ is defined by $G(v) = H(k\concat v)$.
        \end{itemize}

        \cref{lem:qrom-prf} implies
        \[
            \left\{\left(\Hyb_1, O_{H'_1}\right) : H\gets \mathsf{Func}(\cX, \cY)\right\} 
            \approx_c 
            \left\{\left(\Hyb'_1, O_{H'_1} \right) : H\gets \mathsf{Func}(\cX, \cY)\right\}
        \]
        Observe that $\Hyb_{1,0,1}$ makes only syntactic changes from $\Hyb'_1$, i.e. for all $H$,
        \[
            \Hyb_1' = \Hyb_{1,0,1}
        \]
        The same observation holds for every $\Hyb_{1,i,2}$ and $\Hyb_{1,i+1,1}$, i.e. for all $H$,
        \[
            \Hyb_{1,i,2} = \Hyb_{1,i+1,1}
        \]
        
        We now show that\footnote{As before, $H_{1,i}$ is defined using $s_1^{a_i}$, so modifying it to become simulated changes $H'_{1,i}$ to $H'_{1,i+1}$.} 
        \[
            \{(\Hyb_{1,i,1}, O_{H_{1,i}'}): H \gets \mathsf{Func}(\cX, \cY)\} 
            \approx_c^{2^{-\secpar}}
            \{(\Hyb_{1,i+1,2}, O_{H_{1,i+1}'}): H \gets \mathsf{Func}(\cX, \cY)\}
        \]
        by reducing to the $2^{-\secpar}$-security of $\Sigma$. The reduction simulates the random oracles $H$ and $G_{a_i}$ using the compressed oracle technique~\cite{C:Zhandry19},\footnote{If $\Sigma$ is HVZK in the QROM, e.g. because it is designed for the QROM, then this step is unnecessary.} then declares $H(k_\ch \concat a_i)$ as its challenge $s_2^{a_i}$ in an honest-verifier execution of $\Sigma$. 
        It receives either $(s_1^{a_i}, s_3^{a_i})$ generated by an honest prover, or $(\widetilde{s_1^{a_i}}, \widetilde{s_3^{a_i}})$ generated by the HVZK simulator $\Sim_\Sigma$. 
        Using these, it can compute the rest of the experiment according to the description of $\Hyb_{1,i,1}$. In the former case, the resulting distribution is $\Hyb_{1,i,1}$; in the latter, it is $\Hyb_{1,i,2}$. Therefore, if the two distributions above were distinguishable with advantage $\epsilon$, then applying that same distinguisher would distinguish an honest $\Sigma$ prover's messages from the output of the honest-verifier simulator $\Sim_\Sigma$ also with advantage $\epsilon$. The $2^{-\secpar}$-security of $\Sigma$ implies that $\epsilon \leq 2^{-\secpar}$ for all QPT distinguishers.

        \cref{lem:qrom-prf} implies
        \[
            \left\{\left(\Hyb_{1,i_{\max}, 2}, O_{H'_{1,i+1}}\right) : H\gets \mathsf{Func}(\cX, \cY)\right\} 
            \approx_c 
            \left\{\left(\Hyb_2, O_{H'} \right) : H\gets \mathsf{Func}(\cX, \cY)\right\}
        \]

        By the triangle inequality, the distinguishing advantage between $\left\{\left(\Hyb_1, O_{H'_1} \right) : H\gets \mathsf{Func}(\cX, \cY)\right\}$ and $\left\{\left(\Hyb_1, O_{H'_2} \right) : H\gets \mathsf{Func}(\cX, \cY)\right\}$ for any QPT distinguisher is bounded by $\negl(\secpar) + 2^{\secpar/2}\cdot 2^{-\secpar} + \negl(\secpar) = \negl(\secpar)$.

    \ifsubmission\qed\else\fi\end{proof}

    \begin{claim}\label{claim:nizk-hyb2-2}
        \[
        \left\{\left(\Hyb_2, O_{H'} \right) : H\gets \mathsf{Func}(\cX, \cY)\right\}
        \approx 
        \left\{\left(\Hyb_2, O_{H} \right) : H\gets \mathsf{Func}(\cX, \cY)\right\} 
        \]
    \end{claim}
    \begin{proof}
        We reduce to the direct product hardness of subspace states (\cref{lem:subspace-dph}). Assume, for the sake of contradiction, that some QPT distinguisher $D$ with auxiliary input register $\regi{D}$ were able to distinguish between these two distributions. 
        

        \justin{It might be nice to clean this writing up a little and break things down.}
        Before describing the reduction, we claim that $D$ has noticeable advantage conditioned on $\Hyb_2$ not outputting $\bot$, i.e. conditioned on the simulator's certificate check passing. We also claim that $\Hyb_2$ outputs $\bot$ with at most $1-1/p$ probability for some polynomial $p$. We reduce these claims to direct product hardness (\cref{lem:subspace-dph}). The reduction takes in a subspace state $\ket{A}$ and membership oracles $O_A$, $O_{A^\perp}$. If either of the claimed conditions do not hold, then $D$ must have noticeable advantage conditioned on $\Hyb_2$ outputting $\bot$. By \cref{lem:QROM-replacement}, measuring one of $D$'s queries at random produces a point $x^*$ where $H(x^*) \neq H'(x^*)$ with noticeable probability. These are exactly values of the form $a \concat b$ for $a\in A\backslash\{0\}$. Since $H'$ can be implemented using $O_A$, this procedure finds an element of $A$ with noticeable probability using only a polynomial number of queries to $O_A$, without using $\ket{A}$. Combining this with a measurement of $\ket{A}$ in the Hadamard basis, the reduction obtains a pair of vectors in $A\backslash\{0\}\times A^\perp \backslash\{0\}$ with noticeable probability, breaking direct product hardness.
        
        Next, we show how to break direct product hardness of subspace states using a $D$ which has noticeable distinguishing advantage between $(\Hyb_2, O_{H'})$ and $(\Hyb_2, O_{H})$, conditioned on $\Hyb_2$ not outputting $\bot$.
        The reduction takes as input a random subspace state $\ket{A}$ and oracle access to membership oracles $O_{A}$, $O_{A^\perp}$. 
        It runs the simulator (i.e. $\Hyb_2$) using $\ket{A}$. Whenever it would check that a vector $v$ is in $A$ (respectively, $A^\perp$), it queries $v$ to $O_A$ (respectively, $O_{A^\perp}$). To implement step 6 (the certificate check), it first applies the following operations to register $\regi{A}$: (1) a Hadamard operation, (2) the isometry $\ket{y}\mapsto\ket{y-s}$, and (3) a Hadamard operation. Then, it uses $O_A$ and $O_{A^\perp}$ to implement the PVM $\{\ketbra{A}, I - \ketbra{A}\}$ on register $\regi{A}$, as described by \cref{lem:subspace-proj}. If the result of the simulator is $\regi{\adv} \neq \bot$, then it runs the distinguisher $D^{H'}(\regi{D}, \regi{\adv})$ and measures a random query that $D$ makes to $H'$. Let $q = q_1\concat q_2$ be the query measurement result. Finally, the reduction measures $\regi{A}$ in the Hadamard basis to obtain a value $v$ and outputs $(q_1, v)$.

        Recall that $\Hyb_2$ outputs $\regi{\adv}\neq \bot$ with noticeable probability. Condition on this case occurring. This happens exactly when the measurement on $\regi{A}$ returns result $\ketbra{A}$, in which case $\regi{A}$ collapses to $\ket{A}$. Therefore the measurement result $v\in A^\perp$ and is not $0$ with overwhelming probablity. Furthermore, we previously established that $D$ has noticeable advantage in this case; therefore \cref{lem:QROM-replacement} implies that $H(q_1\concat q_2) \neq H'(q_1\concat q_2)$ with noticeable probability. Whenever this occurs, $q_1 \in A$. Therefore $(q_1, v)\in A\backslash\{0\} \times A^\perp\backslash\{0\}$ with noticeable probability, violating \cref{lem:subspace-dph}.
    \ifsubmission\qed\else\fi\end{proof}
    
\ifsubmission\ \qed\else\fi\end{proof}

    \section{Negative Results}

In this section, we give evidence that certifiable deniability is in fact \emph{impossible} in the plain model. Specifically, we show a black-box barrier against obtaining signatures with certified deniability in the plain model; any work achieving this must use non-black-box techniques in their security proof.

\subsection{Supporting Lemma.}

Before proving the main result of this section, we introduce a supporting lemma that generalizes a result from \cite{bbbv97} about reprogramming quantum-accessing oracles.
This lemma may be of independent interest.

In their original version, they consider a quantum adversary who has quantum query access to one of two classical oracles $H$ and $H'$. 
They bound the ability of the adversary to distinguish between the two oracles in terms of the amplitude with which it queries on (classical) inputs $x$ where $H(x) \neq H'(x)$. As a simple corollary, if the adversary is able to distinguish the two oracles in a polynomial number of queries, then measuring one of its queries at random produces an $x$ such that $H_0(x) \neq H_1(x)$ with noticeable probability. 

We observe that a similar statement holds for quantum oracles, i.e. oracles which take in a quantum input, perform quantum computation, and produce quantum output. If the adversary is able to distinguish the two oracles in a polynomial number of queries, then outputting the query register at a random query produces a mixed state with noticeable probability mass on states $\ket{\psi}$ where $H(\ket{\psi})$ is far from $H'(\ket{\psi})$.

\begin{lemma}\label{lem:gen-owth}
    Let $H$ and $H'$ be oracle-accessible unitaries and let $A^{(\cdot)}$ be a quantum oracle algorithm with auxiliary input $\ket{\psi_0}$. Let 
    \[
        \ket{\psi_t} = \sum_{i} \alpha_{t,i} \ket{\phi_{t,i}}_{\regi{A}} \otimes \ket{q_{t, i}}_{\regi{Q}}
    \]
    be a Schmidt decomposition of the state of $A^{H}$ on submitting query $t$, where $\regi{A}$ is the internal register of $A$ and $\regi{Q}$ is the query register. Let $\psi'_t$ similarly be the state of $A^{H'}$ on submitting query $t$. 
    Then for all $T\in \bbN$,
    \[
        \tracedist[\ketbra{\psi_{T+1}}, \ketbra{\psi_{T+1}'}] 
        \leq
        \sqrt{4T \sum_{t=1}^{T} \sum_{i} |\alpha_{t,i}|^2 \tracedist\left[H(\ket{q_{t,i}}), H'(\ket{q_{t,i}})\right]^2}
    \]
\end{lemma} 
\begin{proof}
    The oracle algorithm $A^H$ can be expressed as sequence of unitaries $A_t$ interleaved with unitary oracle operations $H$, i.e.
    \[
        \ket{\psi_{T+1}} = A_{T+1} \left(\prod_{t=1}^{T} H A_t\right) \ket{\psi_0}
    \]
    Note that $\ket{\psi_{t+1}}$ is prepared by applying $A_{t+1}$ to the result $H\ket{\psi_t}$ of the prior query.
    Define $\ket{E_t}$ to be the error term resulting from answering the $t$'th query of $A^H$ using $H'$ instead of $H$, i.e.
    \[
        \ket{E_t} = H'\ket{\psi_{t}} - H\ket{\psi_t}
    \]
    Then we can express the state of $A^H(\ket{\psi_0})$ on submitting the $(T+1)$'th query (before it is answered) as
    \begin{align}
        \ket{\psi_{T+1}} 
        &=  A_T H \ket{\psi_{T}}
        \\
        &= A_T(H'\ket{\psi_{T}} - \ket{E_{T}})
    \end{align}
    Applying this decomposition recursively on $\ket{\psi_{T}}$ yields
    \begin{align}
        \ket{\psi_{T+1}} 
        &= A_{T+1} \left(\prod_{t=1}^{T} H' A_{t} \right) \ket{\psi_{0}} - \sum_{t=1}^{T}\left(\prod_{i=t+1}^{T}  H A_i\right) A_t \ket{E_t}
        \\
        &= \ket{\psi'_{T+1}} - \sum_{t=1}^{T}\left(\prod_{i=t+1}^{T} H A_i \right) \ket{E_{t}}
    \end{align}
    Thus we have
    \allowdisplaybreaks
    \begin{align}
        &\tracedist[\ketbra{\psi_{T+1}}, \ketbra{\psi_{T+1}'}]\nonumber
        \\
        &= \sqrt{1-|\braket{\psi_T}{\psi'_T}|^2}
        \\
        &\leq \sqrt{1-|\braket{\psi_T} +(\bra{\psi_T'} - \bra{\psi_T})\ket{\psi_T}|^2}
        \\
        &\leq \sqrt{|(\bra{\psi_T'} - \bra{\psi_T})\ket{\psi_T}|^2}
        \\
        &\leq \sqrt{\left\|\ket{\psi_T'} - \ket{\psi_T} \right\|_2^2}
        \label{eq:gen-owtf-cs}
        \\
        &= \sqrt{\left\|\sum_{t=1}^{T}\left(\prod_{i=t+1}^{T} H A_i \right) \ket{E_{t}}\right\|_2^2}
        \\
        &\leq \sqrt{T \sum_{t=1}^{T}\left\|\left(\prod_{i=t+1}^{T} H A_i \right) \ket{E_{t}}\right\|_2^2}
        \label{eq:owth-general-jensen}
        \\
        &= \sqrt{T \sum_{t=1}^{T}\left\| \ket{E_{t}}\right\|_2^2}
        \label{eq:owth-general-unitary}
        \\
        &=  \sqrt{T \sum_{t=1}^{T}\left\| H'\ket{\psi_{t}} - H\ket{\psi_t}\right\|_2^2}
        \\
        &= \sqrt{T \sum_{t=1}^{T}\left| \sum_{i,j} \alpha_{t,i}\overline{\alpha_t,j} \braket{\phi_{t,i}}{\phi_{t,j}}    \bra{q_{t,i}}(H-H')^\dagger(H-H')\ket{q_{t,j}}\right|^2}
        \\
        &\leq \sqrt{T \sum_{t=1}^{T} \sum_{i} |\alpha_{t,i}|^2 \left\|(H-H')\ket{q_{t,i}}\right\|_2^2}
        \label{eq:gen-owth-jensen2}
        \\
        &\leq \sqrt{T \sum_{t=1}^{T} \sum_{i} |\alpha_{t,i}|^2 \left\|(H-H')\ket{q_{t,i}}\right\|_1^2}
        \\
        &= \sqrt{4T \sum_{t=1}^{T} \sum_{i} |\alpha_{t,i}|^2 \tracedist\left[H(\ket{q_{t,i}}), H'\ket{q_{t,i}}\right]^2}
    \end{align}
    where \cref{eq:gen-owtf-cs} follows from Cauchy-Schwarz; \cref{eq:owth-general-jensen} follows from Jensen's inequality; \cref{eq:owth-general-unitary} follows from the invariance of the $\ell_2$ norm under unitaries; \cref{eq:gen-owth-jensen2} follows from Jensen's inequality.
\end{proof}

\begin{corollary}[One-Way to Hiding for Unitary Oracles]\label{coro:gen-owth}
    Let $H$ and $H'$ be oracle-accessible unitaries and let $A^{(\cdot)}$ be a quantum oracle algorithm with auxiliary input $\ket{\psi_0}$ that uses at most $T$ queries. Denote the mixed state resulting from measuring $A^H$'s query register at a random time $t$ as
    \[
        \frac{1}{T}\sum_{t=1}^{T} \ketbra{t} \otimes \sum_i |\alpha_{t,i}|^2 \ketbra{q_{t,i}}
    \]
    For $0< \delta \leq 1$, denote the set $S_\delta$ as being the subset of $(t,i)$ pairs such that $\tracedist[H(\ket{q_{t,i}}), H'(\ket{q_{t,i}})] \geq \delta$. If $A^{(\cdot)}$ distinguishes between $H$ and $H'$ with advantage $\epsilon$ in $T$ queries, then
    \[
        \frac{1}{T}\sum_{(t,i) \in S_\delta} |\alpha_{t,i}|^2 \geq \frac{\frac{\epsilon^2}{4T^2} - \delta^2}{1-\delta^2}
    \]
\end{corollary}

If $\epsilon$ were noticeable, $T$ were polynomial, and $\delta$ were set to $(1/2)\cdot \epsilon/(2T)$, then this immediately shows that the probability mass on pure states where $H$ and $H'$ differ by a noticeable trace distance is noticeable.


\subsection{Plain Model Black-Box Barrier}\label{sec:plainmodimp}

\begin{theorem}
    There do not exist signatures with certified deniability in the plain model whose security proof makes black-box use of the adversary.
\end{theorem}
\begin{remark}
    We note that a similar statement holds for NIZKs with certified deniability as well. The proof is almost identical, so we omit it here.
\end{remark}
\begin{proof}
    Let $\sigcd = (\Gen,\allowbreak \Sign,\allowbreak \Verify,\allowbreak \Del,\allowbreak \DelVer)$ be a candidate signature scheme with a black-box proof of certified deniability, i.e. which gives a simulator $\Sim$ that uses the adversary as a black-box. We give a QPT adversary $\Adv$ together with a QPT distinguisher $D$ that distinguishes $\Sim$ from the real experiment. 
    
    $\Adv$ receives as auxiliary input a post-quantum signing key $\widetilde{k}$, which we denote as $\Adv_{\widetilde{k}}$. 
    During the certified deniability game, $\Adv_{\widetilde{k}}$ receives a verification key $\verkey$, a state $\ket{\sigma}$, and a message $m$, then runs $\Verify(\verkey \ket{\sigma}, m)$, and if it accepts, outputs a post-quantum signature $\sigma^{(\widetilde{k})}_{m,\verkey} \gets \PQSig.\Sign(\widetilde{k}, \verkey\concat m)$ on the verification key concatenated with the message. 
    It measures the result of this to get $\sigma^{(\widetilde{k})}_{m,\verkey}$, which it writes to register $\regi{\adv}$, then uncomputes the program. By the correctness of $\sigcd$, this is gentle. 
    Finally, it runs $\Del(\ket{\sigma_m})$ to obtain a valid deletion certificate in register $\regi{\cert}$ and outputs $\regi{\cert}$ along with $\regi{\adv}$.

    The distinguisher receives as auxiliary input the corresponding post-quantum verification key $\widetilde{\verkey}$. On input $\regi{\adv}$, it runs the post-quantum signature verification $\PQSig.\Verify(\widetilde{\verkey}, \regi{\adv})$ to check if $\adv$ obtained a signature on $\verkey \concat m$. If so, it outputs $1$, and otherwise, outputs $0$. 

    We now show that this adversary and distinguisher pair achieves noticeable advantage against any simulator $\Sim$. 
    Observe that in the real certified deniability experiment, $D$ outputs $1$ with overwhelming probability, due to the correctness of the signature and obfuscation schemes. 
    Conversely, we claim that against any QPT simulator, $D$ outputs $0$ with overwhelming probability. 
    
    Consider the following sequence of hybrid experiments:
    \begin{itemize}
        \item \underline{$\Hyb_0$:} The simulated certified deniability experiment. In more detail, sample $(\sk, \verkey)\gets \Gen(1^\secpar)$. Run $\Sim^{\Adv_{\widetilde{k}}}$ to obtain an output register $\regi{\adv}$.

        \item \underline{$\Hyb_1$:} This experiment is the same as $\Hyb_0$, except that $\Adv_{\widetilde{k}}$ is replaced by $\Adv_\bot$, which immediately honestly deletes the signature and always outputs $\bot$ in register $\regi{\adv}$.
    \end{itemize}
    
    $\Hyb_0 \approx \Hyb_1$ from the unforgeability of the candidate signature scheme; otherwise, \cref{coro:gen-owth} implies that outputting $\Sim^{\Adv_{\widetilde{k}}}$'s query register on a random query produces a mixed state with noticeable probability mass on an input $\ket{\psi}$ where $\adv_{\widetilde{k}}(\ket{\psi})$ and $\adv_\bot(\ket{\psi})$ are noticeably far, which is is exactly the set of points where $\ket{\psi}$ passes signature verification with noticeable probability.

    We now claim that the original distinguisher $D(\widetilde{\verkey})$ outputs $0$ with overwhelming probability in $\Hyb_1$. Otherwise, $\Sim^{\adv_{\bot}}$ must output a valid signature under $\widetilde{\verkey}$, which it is independent of, which would violate the unforgeability of the post-quantum signature scheme. 
    Since $\Hyb_1$ is indistinguishable from $\Hyb_0$ for all QPT distinguishers, $D(\widetilde{\verkey})$ must also output $0$ with overwhelming probability in $\Hyb_0$, which is simply the output distribution of the certified deniability simulator.

\end{proof}

\paragraph{Obfuscated Auxiliary Input.} One could also consider implementing the adversary described above using an obfuscated program which $\adv$ receives as auxiliary input. In this case, it might not be possible to extract a signature from the obfuscated program even using non-black-box techniques.

\section{Acknowledgements}
Alper Çakan and Justin Raizes were supported by the following grants of Vipul Goyal: NSF award 1916939, DARPA SIEVE program, a gift from Ripple, a DoE NETL award, a JP Morgan Faculty Fellowship, a PNC center for financial services innovation award, and a Cylab seed funding award.

\fi

\bibliographystyle{alpha}
\bibliography{refs,abbrev2,crypto}

\ifsubmission
\appendix

\else \fi

\end{document}